\PassOptionsToPackage{unicode}{hyperref}
\PassOptionsToPackage{hyphens}{url}
\PassOptionsToPackage{dvipsnames,svgnames*,x11names*}{xcolor}
\documentclass[
  12pt,
]{article}
\usepackage{lmodern}
\usepackage{amssymb,amsmath}
\usepackage{ifxetex,ifluatex}
\ifnum 0\ifxetex 1\fi\ifluatex 1\fi=0 
  \usepackage[T1]{fontenc}
  \usepackage[utf8]{inputenc}
  \usepackage{textcomp} 
\else 
  \usepackage{unicode-math}
  \defaultfontfeatures{Scale=MatchLowercase}
  \defaultfontfeatures[\rmfamily]{Ligatures=TeX,Scale=1}
\fi
\IfFileExists{upquote.sty}{\usepackage{upquote}}{}
\IfFileExists{microtype.sty}{
  \usepackage[]{microtype}
  \UseMicrotypeSet[protrusion]{basicmath} 
}{}
\makeatletter
\@ifundefined{KOMAClassName}{
  \IfFileExists{parskip.sty}{%
    \usepackage{parskip}
  }{
    \setlength{\parindent}{0pt}
    \setlength{\parskip}{6pt plus 2pt minus 1pt}}
}{
  \KOMAoptions{parskip=half}}
\makeatother
\usepackage{xcolor}
\IfFileExists{xurl.sty}{\usepackage{xurl}}{} 
\IfFileExists{bookmark.sty}{\usepackage{bookmark}}{\usepackage{hyperref}}
\hypersetup{
  pdftitle={Pricing and Capital Allocation for Multiline Insurance Firms With Finite Assets in an Imperfect Market},
  pdfauthor={John A. Major; Stephen J. Mildenhall},
  colorlinks=true,
  linkcolor=Maroon,
  filecolor=Maroon,
  citecolor=Blue,
  urlcolor=cyan,
  pdfcreator={LaTeX via pandoc}}
\usepackage[margin=1in]{geometry}
\usepackage{longtable,booktabs}
\usepackage{etoolbox}
\makeatletter
\patchcmd\longtable{\par}{\if@noskipsec\mbox{}\fi\par}{}{}
\makeatother
\IfFileExists{footnotehyper.sty}{\usepackage{footnotehyper}}{\usepackage{footnote}}
\makesavenoteenv{longtable}
\usepackage{graphicx}
\makeatletter
\def\maxwidth{\ifdim\Gin@nat@width>\linewidth\linewidth\else\Gin@nat@width\fi}
\def\maxheight{\ifdim\Gin@nat@height>\textheight\textheight\else\Gin@nat@height\fi}
\makeatother
\setkeys{Gin}{width=\maxwidth,height=\maxheight,keepaspectratio}
\makeatletter
\def\fps@figure{htbp}
\makeatother
\providecommand{\tightlist}{%
  \setlength{\itemsep}{0pt}\setlength{\parskip}{0pt}}
\usepackage{accents}
\usepackage{mathtools}
\usepackage{rotating}
\usepackage{ifthen}
\usepackage{bm}
\usepackage{xcolor}
\usepackage{amsmath}
\usepackage{amsthm}
\usepackage{empheq}
\usepackage{dcolumn}  
\usepackage{mathrsfs}
\usepackage{tikz}
\usepackage{etoolbox}
\usepackage{pgfplots}
\usepackage{pgfplotstable}
\usepackage{verbatim}
\usepackage{tikz-cd}
\usepackage{longtable}
\usepackage{cleveref}
\usepackage{afterpage}
\usepackage{ctable} 
\usepackage{centernot}

\usepackage{subfig}

\providecommand{\subtitle}[1]{%
  \usepackage{titling}
  \posttitle{%
    \par\large#1\end{center}}
}

\def\E{\mathsf{E}}

\def\var{\mathsf{var}}

\def\TVaR{\mathsf{TVaR}}
\def\cov{\mathsf{cov}}

\def\Pr{\mathsf{Pr}}

\def\dint{\displaystyle\int}

\newtheorem{theorem}{Theorem}

\newtheorem{proposition}{Proposition}
\newtheorem{corollary}{Corollary}
\newtheorem{definition}{Definition}

\pgfplotsset{width=7cm,compat=1.8}
\usetikzlibrary{arrows,calc,positioning,shadows.blur,decorations.pathreplacing}
\usetikzlibrary{automata}
\usetikzlibrary{fit}
\usetikzlibrary{snakes}
\usetikzlibrary{intersections}
\usetikzlibrary{decorations.markings,decorations.text,decorations.pathmorphing,decorations.shapes}
\usetikzlibrary{decorations.fractals,decorations.footprints}
\usetikzlibrary{graphs}
\usetikzlibrary{matrix}
\usetikzlibrary{shapes.geometric}
\usetikzlibrary{mindmap, shadows}
\usetikzlibrary {backgrounds}

\newlength{\cslhangindent}
\newenvironment{cslreferences}%
  {}%
  {\par}

\title{Pricing and Capital Allocation for Multiline Insurance Firms With
Finite Assets in an Imperfect Market}
\author{John A. Major \and Stephen J. Mildenhall}
\date{Created 2020-08-27 12:24:17.937134}

\begin{document}
\maketitle
\begin{abstract}
We analyze multiline pricing and capital allocation in equilibrium
no-arbitrage markets. Existing theories often assume a perfect complete
market, but when pricing is linear, there is no diversification benefit
from risk pooling and therefore no role for insurance companies. Instead
of a perfect market, we assume a non-additive distortion pricing
functional and the principle of equal priority of payments in default.
Under these assumptions, we derive a canonical allocation of premium and
margin, with properties that merit the name the natural allocation. The
natural allocation gives non-negative margins to all independent lines
for default-free insurance but can exhibit negative margins for low-risk
lines under limited liability. We introduce novel conditional
expectation measures of relative risk within a portfolio and use them to
derive simple, intuitively appealing expressions for risk margins and
capital allocations. We give a unique capital allocation consistent with
our law invariant pricing functional. Such allocations produce returns
that vary by line, in contrast to many other approaches. Our model
provides a bridge between the theoretical perspective that there should
be no compensation for bearing diversifiable risk and the empirical
observation that more risky lines fetch higher margins relative to
subjective expected values.

\end{abstract}

\hypertarget{introduction}{%
\section{Introduction}\label{introduction}}

The complete perfect market paradigm dominates financial models of
insurance pricing developed since the 1970s. These models imply the
existence of an additive valuation rule, meaning there is no benefit
from diversification and no role for insurer intermediaries. Moreover,
they typically provide no compensation for bearing diversifiable risk.
Milestones in this approach include the use of discounted cash flows,
Myers and Cohn (\protect\hyperlink{ref-Myers1987}{1987}) and Cummins
(\protect\hyperlink{ref-Cummins1990a}{1990}), and options pricing
Doherty and Garven (\protect\hyperlink{ref-Doherty1986}{1986}) and
Cummins (\protect\hyperlink{ref-Cummins1988}{1988}). Phillips, Cummins,
and Allen (\protect\hyperlink{ref-Phillips1998}{1998}) and Myers and
Read Jr. (\protect\hyperlink{ref-Myers2001}{2001}) considered multiple
line pricing, allowing for ex ante default rules, which Sherris
(\protect\hyperlink{ref-Sherris2006a}{2006}) and Ibragimov, Jaffee, and
Walden (\protect\hyperlink{ref-Ibragimov2010}{2010}) extended to ex post
rules. From Phillips's contribution forward, there is a focus on what
happens in default states to explain premiums by line. Aside from
non-diversifiable risk and tail-driven default, these models do not
reflect the overall distribution of losses, or what we will call the
\emph{shape of risk}.

Holding capital in an intermediary insurance company introduces
frictional costs, which Cummins
(\protect\hyperlink{ref-Cummins2000}{2000}) explains are primarily
driven by agency conflict, taxes, and regulation. Why would insureds
incur the extra expense of buying through an insurer in a perfect
market? Ibragimov, Jaffee, and Walden
(\protect\hyperlink{ref-Ibragimov2010}{2010}) assume insureds do not
have direct access to ultimate capital providers or that there are
additional costs to do so. This somewhat unsatisfactory explanation was
anticipated by Grundl and Schmeiser
(\protect\hyperlink{ref-Grundl2007}{2007}), who pointed out that pricing
does not depend on allocation in a perfect market. Bauer and Zanjani
(\protect\hyperlink{ref-Bauer2013}{2013}) consider capital allocation in
the context of allocating frictional costs.

The fundamental theorem of asset pricing states that in a perfect market
the existence of an additive valuation rule is essentially equivalent to
no arbitrage, Ross (\protect\hyperlink{ref-Ross1978}{1978}), Dybvig and
Ross (\protect\hyperlink{ref-Dybvig1989}{1989}), Delbaen and
Schachermayer (\protect\hyperlink{ref-Delbaen1994}{1994}). The thought
that insurance valuation rules must be additive has exerted a
considerable influence on theory and practice, Borch
(\protect\hyperlink{ref-Borch1982}{1982}), Venter
(\protect\hyperlink{ref-Venter1991}{1991}). But while indubitably
competitive, the insurance market is neither perfect nor complete.

Imperfect or incomplete pricing paradigms isolate different failures of
the perfect complete model but tend to have very similar implications.
They result in non-additive pricing functionals that are often additive
on comonotonic risks, and that can be expressed as a worst-case
expectation over a set of probability distribution outcomes. Wang
(\protect\hyperlink{ref-Wang1996}{1996}) and Wang, Young, and Panjer
(\protect\hyperlink{ref-Wang1997}{1997}) apply a non-additive functional
using distorted probabilities to insurance pricing, leveraging diverse
theoretical underpinnings including Huber
(\protect\hyperlink{ref-Huber1981}{1981}), Schmeidler
(\protect\hyperlink{ref-Schmeidler1986}{1986}), Schmeidler
(\protect\hyperlink{ref-Schmeidler1989}{1989}), Yaari
(\protect\hyperlink{ref-Yaari1987}{1987}), and Denneberg
(\protect\hyperlink{ref-Denneberg1994}{1994}). Within this theory,
distortion risk measures (DRM) occur repeatedly and in many guises.
Kusuoka (\protect\hyperlink{ref-Kusuoka2001}{2001}) and Acerbi
(\protect\hyperlink{ref-Acerbi2002b}{2002}) characterize DRMs as
coherent, law invariant, and comonotonic additive functionals.

DRMs are easy to apply in practice and have many appealing properties.
However, they are not additive and include transaction costs, via an
implied bid-ask spread, Castagnoli, Maccheroni, and Marinacci
(\protect\hyperlink{ref-Castagnoli2004}{2004}). We must ask whether they
are consistent with arbitrage-free pricing. Fortunately, the answer is
yes. The presence of transaction costs render apparent arbitrage
opportunities impractical and so a non-additive pricing rule can still
be arbitrage-free.

Results from Chateauneuf, Kast, and Lapied
(\protect\hyperlink{ref-Chateauneuf1996a}{1996}), De Waegenaere
(\protect\hyperlink{ref-DeWaegenaere2000}{2000}), and especially
Castagnoli, Maccheroni, and Marinacci
(\protect\hyperlink{ref-Castagnoli2002}{2002}) and De Waegenaere, Kast,
and Lapied (\protect\hyperlink{ref-DeWaegenaere2003}{2003}) produce
general equilibrium models that allow for non-additive prices, and show
that DRM pricing---in the equivalent guise as a Choquet integral---is
consistent with general equilibrium. As a result, it is legitimate to
use DRMs to price risk transfer. They return a premium that charges for
the shape of risk, even diversifiable risk. DRMs can be applied both to
non-intermediated, direct-to-investor pricing and insurer intermediated
pricing, where insurers sell the risk to investors.

Why is it reasonable to charge for the shape of risk, especially when it
is diversifiable? There are two related reasons: ambiguity aversion and
winner's curse.

Insurance pricing is a horse lottery, not a roulette lottery; it relies
on subjective probabilities Anscombe and Aumann
(\protect\hyperlink{ref-Anscombe1963}{1963}). Losses are ambiguous and
investors are ambiguity averse. Zhang
(\protect\hyperlink{ref-Zhang2002}{2002}) and Klibanoff, Marinacci, and
Mukerji (\protect\hyperlink{ref-Klibanoff2005}{2005}) describe ambiguity
relevant to insurance pricing. The latter paper has been applied in an
insurance context by Robert and Therond
(\protect\hyperlink{ref-Robert2014}{2014}), Dietz and Walker
(\protect\hyperlink{ref-Dietz2017}{2017}) and Jiang, Escobar-Anel, and
Ren (\protect\hyperlink{ref-Jiang2020}{2020}). Epstein and Schneider
(\protect\hyperlink{ref-Epstein2008}{2008}) can be applied equally to
underwriters who often weight bad news more heavily than good. Premium
ambiguity is confounded with risk because more risky classes of business
tend to have more ambiguous premiums. The failure of the terrorism
insurance market is a case in point.

The second reason is the winner's curse: the winning bid is biased low
when there are multiple quotes, Thaler
(\protect\hyperlink{ref-Thaler1998}{1988}). The insurance market is very
competitive and is characterized by big-data, predictive modeling
pricing with heterogeneous insureds. Competing classification plans
exacerbate winner's curse. As a result, a margin over subjective
expected loss, even when subjective probabilities are unbiased, is
justified. Winner's curse will be correlated to ambiguity: greater
ambiguity will result in wider quote dispersion. The winner's curse
margin will appear in quotes but will not appear ex post in results. The
winner's curse is distinct from adverse selection, D'Arcy and Doherty
(\protect\hyperlink{ref-Doherty1990a}{1990})

In this paper, we consider applying a DRM pricing functional to
individual policies. We do this in two steps. First, we consider pricing
without insurers and associated frictional costs. In this world,
insureds contract directly with investor capital providers. The capital
providers can be risk-neutral, but they are ambiguity averse. Uncertain
subjective probabilities drive the market; there are no objective
probabilities. Investors may be willing to write a roulette lottery at
cost, but they are unwilling to write a horse lottery at cost because
they are ambiguity averse and because they know their winning quotes
will be biased below subjective expectation. The DRM incorporates a
margin to allow for ambiguity aversion and help correct for the winner's
curse.

Non-additive DRMs create a benefit to independent insureds who pool
their risks. Estimation risk, process risk, and entropy are all lower
for the pool than for the individual risks, resulting in less ambiguous
subjective probabilities. Pool pricing is closer to the subjective
expected value, lowering the average premium for pool participants.
These conclusions do not hold if the insureds are not independent, for
example if the risk is driven by catastrophe events. Boonen, Tsanakas,
and Wüthrich (\protect\hyperlink{ref-Tsanakas2016a}{2017}) and
Mildenhall (\protect\hyperlink{ref-Mildenhall2017b}{2017}) discuss the
important differences between the independent and dependent cases.

Pool members still have the problem of allocating any gains from
diversification. We show there is a unique way to do this consistent
with the value of insurance cash flows and an assumption of equal
priority in default, but relying on no other inputs. We call our method
the natural allocation. Our approach is similar to Ibragimov, Jaffee,
and Walden (\protect\hyperlink{ref-Ibragimov2010}{2010}). We model fair
value to insureds rather than marginal costs to the pool. The pool is a
transparent, contractual pass-through.

Our method identifies the conditional expectation of policy losses given
total pool loss as a controlling function, which we call \(\kappa\). It
is exquisitely sensitive to relative risk within a portfolio. We explain
our allocation in reference to the relative consumption of more or less
ambiguous, and hence costly, capital layers.

Risk pools have a clear role in our model: they minimize ambiguity-based
risk costs. Pools accumulate risk to create more credible pricing
signals, credible being synonymous with stable and low volatility.
Stable pools will be financed more cheaply by investors than more
ambiguous pools, or than single risks. A similar data-centric role for
insurance pools was suggested in Froot and O'Connell
(\protect\hyperlink{ref-Froot2008}{2008}).

A risk pool does not have to be structured as an insurance company
intermediary. Intermediaries emerge as pool managers look to bolster
their credibility with investors by putting skin in the game and
assuming risk. This provides a solid rationale for the existence of
insurers versus pools managed by non-risk bearing underwriting managers.
Interestingly, non-risk-bearing managers are now quite common in the
property catastrophe reinsurance market. There, credibility is enhanced
by reliance on independent and highly regarded third-party catastrophe
models. Such models do not exist for most lines. Generally, aggregated
data is needed to price individual policies fairly.

Insurance company pools are successful if they lower the cost of pure
risk transfer by an amount that more than offsets their frictional costs
of holding capital. Individual policy pricing still requires an
allocation of frictional costs. Since frictional costs are usually
regarded as a flat tax on capital, this can be done via a capital
allocation. We show there is a unique ``natural'' way to perform the
allocation, giving a complete solution to pooled risk insurance pricing
under a DRM. Like the premium allocation, the capital allocation only
assumes aggregate pricing by a DRM and equal priority in default. It
relies on the fact that DRMs are law invariant.

Although many of the results in the paper have been known, at least in
principle, for many years, we believe it contains several noteworthy
results. In particular, we highlight the following contributions.

\begin{itemize}
\tightlist
\item
  We analyze multiline pricing and capital allocation in equilibrium
  no-arbitrage markets. We address the allocation of risk measures from
  the insured's perspective, as in Ibragimov, Jaffee, and Walden
  (\protect\hyperlink{ref-Ibragimov2010}{2010}), but generalize their
  work to pricing in imperfect and incomplete markets with non-additive
  functionals.
\item
  We find there is a canonical allocation of premium and margin in a
  market where prices are determined by a distortion risk measure with a
  rule of equal priority in default. The allocation relies on no
  additional assumptions and we believe it merits the name the natural
  allocation. The natural allocation gives non-negative margins to all
  independent lines for default-free insurance but can exhibit negative
  margins for low-risk lines under limited liability. The natural
  allocation reveals subtle interactions between margin, default, and
  idiosyncratic risk.
\item
  We introduce novel conditional expectation measures of relative risk
  within a portfolio and use them to derive simple, intuitively
  appealing expressions for risk margins and capital allocations
  consistent with law invariant DRM pricing functionals.
\item
  We give a unique capital allocation consistent with our law invariant
  pricing functional. Such allocations produce returns that vary by
  line, in contrast to many other approaches.
\item
  We illustrate the theory with examples. The examples elucidate the
  interplay of loss and margin between lines at various levels of
  portfolio risk. We demonstrate how the cost of the shape of risk
  reflects a complex interaction between the relative consumption of low
  layer, certain, high loss ratio assets, and high layer, uncertain, low
  loss ratio assets.
\item
  In contrast with cost allocation-based approaches, we find that margin
  by line is more driven by behavior in solvent states than in default
  states. Our results hold even when there is no possibility of default.
\item
  Our model provides a bridge between the theoretical perspective that
  there should be no compensation for bearing diversifiable risk and the
  empirical observation that more risky lines fetch higher margins
  relative to subjective expected values.
\end{itemize}

The remainder of the paper is structured as follows. Section 2 recalls
the definition and properties of a DRM. Section 3 states and proves
canonical formulas giving the natural allocation of loss and premium by
policy under a DRM, and shows margins are non-negative for independent
risks. Section 4 discusses the \(\kappa\) and two associated functions
that control the natural allocation and that are informative in their
own right. Section 5 derives some properties of the natural margin
allocation by policy. Section 6 extends the natural premium allocation
to a natural allocation of capital. This allows frictional costs to be
allocated. Section 7 gives two examples, one simple and one more
realistic. Finally, Section 8 offers concluding remarks, discusses
limitations, and makes suggestions for further research.

\hypertarget{notation-and-conventions}{%
\subsection*{Notation and Conventions}\label{notation-and-conventions}}
\addcontentsline{toc}{subsection}{Notation and Conventions}

We consider insurance written directly by investors or intermediated by
an insurance company. We distinguish, when necessary, by saying
investor-written, or intermediary-written or intermediated insurance.
When irrelevant, we just say the insurance is written by a provider. In
either case, insurance consists of a pool of individual policies written
for insureds. Policies last one period, with premium collected at
\(t=0\) and losses paid in full at \(t=1\). Line of business or business
unit or other grouping can be substituted for policy. We refer to the
components as lines or policies, whichever is most appropriate.

When insurance is supported by finite assets, the provider has limited
liability.

When intermediated, the intermediary is a stock insurer. At \(t=0\) it
sells its residual value to investors to raise equity. At time one it
pays claims up to the amount of assets available. If assets are
insufficient to pay claims it defaults. If there are excess assets they
are returned to investors.

The terminology describing risk measures is standard, and follows
Föllmer and Schied (\protect\hyperlink{ref-Follmer2011}{2011}). We work
on a standard probability space, Svindland
(\protect\hyperlink{ref-Svindland2010}{2009}), Appendix. It can be taken
as \(\Omega=[0,1]\), with the Borel sigma-algebra and \(\mathsf P\)
Lebesgue measure. The indicator function on a set \(A\) is \(1_A\),
meaning \(1_A(x)=1\) if \(x\in A\) and \(1_A(x)=0\) otherwise.

Total insured loss, or total risk, is described by a random variable
\(X\ge 0\). \(X\) reflects policy limits but is not limited by provider
assets. \(X=\sum_i X_i\) describes the split of losses by policy. \(F\),
\(S\), \(f\), and \(q\) are the distribution, survival, density, and
(lower) quantile functions of \(X\). Subscripts are used to
disambiguate, e.g., \(S_{X_i}\) is the survival function of \(X_i\).
\(X\wedge a\) denotes \(\min(X,a)\) and \(X^+=\max(X,0)\).

The letters \(S\), \(P\), \(M\) and \(Q\) refer to expected loss,
premium, margin and equity, and \(a\) refers to assets. The value of
survival function \(S(x)\) is the loss cost of the insurance paying
\(1_{\{X>x\}}\), so the two uses of \(S\) are consistent. Premium equals
expected loss plus margin; assets equal premium plus equity. All these
quantities are functions of assets underlying the insurance.

We use the actuarial sign convention: large positive values are bad. Our
concern is with quantiles \(q(p)\) for \(p\) near 1. Distortions are
usually reversed, with \(g(s)\) for small \(s=1-p\) corresponding to bad
outcomes. As far as possible we will use \(p\) in the context \(p\)
close to 1 is bad and \(s\) when small \(s\) is bad.

Tail value at risk is defined for \(0\le p<1\) by \[
\TVaR_p(X) = \frac{1}{1-p}\int_p^1 q(t)dt.
\]

Prices exclude all expenses. The risk free interest rate is zero. These
are standard simplifying assumptions, e.g.~Ibragimov, Jaffee, and Walden
(\protect\hyperlink{ref-Ibragimov2010}{2010}).

\hypertarget{distortion-risk-measures-and-pricing-functionals}{%
\section{Distortion Risk Measures and Pricing
Functionals}\label{distortion-risk-measures-and-pricing-functionals}}

We define DRMs and recall results describing their different
representations. By De Waegenaere, Kast, and Lapied
(\protect\hyperlink{ref-DeWaegenaere2003}{2003}) DRMs are consistent
with general equilibrium and so it makes sense to consider them as
pricing functionals. The DRM is interpreted as the (ask) price for an
investor-written risk transfer. The rest of the paper will explore
multi-policy pricing implied by a DRM.

\medskip

\begin{definition} A {\bf distortion function} is an increasing concave function $g:[0,1]\to [0,1]$ satisfying $g(0)=0$ and $g(1)=1$.

A {\bf distortion risk measure} $\rho_g$ associated with a distortion $g$ acts on a non-negative random variable $X$ as
\begin{equation}
\rho_g(X) = \int_0^\infty g(S(x))dx.
\label{eq:drm-def}
\end{equation}
\end{definition}

The simplest distortion if the identity \(g(s)=s\). Then
\(\rho_g(X)=\E[X]\) from the integration-by-parts identity \[
\int_0^\infty S(x)\,dx = \int_0^\infty xdF(x).
\] Other well-known distortions include the proportional hazard
\(g(s)=s^r\) for \(0<r\le 1\) and the Wang transform
\(g(s)=\Phi(\Phi^{-1}(s)+\lambda)\) for \(\lambda \ge 0\), Wang
(\protect\hyperlink{ref-Wang1995}{1995}).

Since \(g\) is concave \(g(s)\ge 0g(0) + sg(1)=s\) for all
\(0\le s\le 1\), showing \(\rho_g\) adds a non-negative margin.

Going forward, \(g\) is a distortion and \(\rho\) is its associated
distortion risk measure. We interpret \(\rho\) as a pricing functional
and refer to \(\rho(X)\) as the price or premium for investor-written
insurance on \(X\). When we price intermediated insurance we need to add
frictional costs of holding capital. This is considered in
\cref{intermediated-pricing}.

DRMs are translation invariant, monotonic, subadditive and positive
homogeneous, and hence coherent, Acerbi
(\protect\hyperlink{ref-Acerbi2002b}{2002}). In addition they are law
invariant and comonotonic additive. In fact, all such functionals are
DRMs. As well has having these properties, DRMs are powerful because we
have a complete understanding of their representation and structure,
which we summarize in the following theorem.

\begin{theorem}\label{thm:coherent}
Subject to $\rho$ satisfying certain continuity assumptions, the following are equivalent.
\begin{enumerate}
\item $\rho$ is a law invariant, coherent, comonotonic additive risk measure.
\item $\rho=\rho_g$ for a concave distortion $g$.
\item $\rho$ has a representation as a weighted average of TVaRs for a measure $\mu$ on $[0,1]$ $\rho(X)=\int_0^1 \TVaR_p(X)\mu(dp)$.
\item $\rho(X)=\max_{\mathsf Q\in\mathscr{Q}} \E_{\mathsf Q}[X]$ where $\mathscr{Q}$ is the set of (finitely) additive measures with $\mathsf Q(A)\le g(\mathsf P(A))$ for all measurable $A$.
\item $\rho(X)=\max_{\mathsf Z\in\mathscr{Z}} \E[XZ]$ where $\mathscr{Z}$ is the set of positive functions on $\Omega$ satisfying $\int_p^1 q_Z(t)dt \le g(1-p)$, and $q_Z$ is the quantile function of $Z$.
\end{enumerate}
\end{theorem}

The theorem combines results from Föllmer and Schied
(\protect\hyperlink{ref-Follmer2011}{2011}) (4.79, 4.80, 4.93, 4.94,
4.95), Delbaen (\protect\hyperlink{ref-Delbaen2000}{2000}), Kusuoka
(\protect\hyperlink{ref-Kusuoka2001}{2001}), and Carlier and Dana
(\protect\hyperlink{ref-Carlier2003}{2003}). The theorem requires that
\(\rho\) is continuous from above to rule out the possibility
\(\rho=\sup\). In certain situations, the \(\sup\) risk measure applied
to an unbounded random variable can only be represented as a \(\sup\)
over a set of test measures and not a max. Note that the roles of from
above and below are swapped from Föllmer and Schied
(\protect\hyperlink{ref-Follmer2011}{2011}) because they use the asset,
negative is bad, sign convention whereas we use the actuarial, positive
is bad, convention.

The relationship between \(\mu\) and \(g\) is given by Föllmer and
Schied (\protect\hyperlink{ref-Follmer2011}{2011}) 4.69 and 4.70. The
elements of \(\mathscr Z\) are the Radon-Nikodym derivatives of the
measures in \(\mathscr Q\).

\hypertarget{loss-and-premium-by-policy-and-layer}{%
\section{Loss and Premium by Policy and
Layer}\label{loss-and-premium-by-policy-and-layer}}

This section introduces the idea of layer densities and proves that DRM
premium can be allocated to policy in a natural and unique way.

\hypertarget{layer-densities}{%
\subsection{Layer Densities}\label{layer-densities}}

Risk is often tranched into layers that are then insured and priced
separately. Meyers (\protect\hyperlink{ref-Meyers1996}{1996}) describes
layering in the context of liability increased limits factors and Culp
and O'Donnell (\protect\hyperlink{ref-Culp2009}{2009}), Mango et al.
(\protect\hyperlink{ref-Mango2013}{2013}) in the context of excess of
loss reinsurance.

Define a layer \(y\) excess \(x\) by its payout function
\(I_{(x,x+y]}(X):=(X-x)^+\wedge y\). The expected layer loss is
\begin{align*}
\E[I_{(x,x+y]}(X)] &= \int_x^{x+y} (t-x)dF(t) + yS(x+y) \\
&= \int_x^{x+y} t dF(t) + tS(t)\vert_x^{x+y} \\
&= \int_x^{x+y} S(t)\, dt.
\end{align*} Based on this equation, Wang
(\protect\hyperlink{ref-Wang1996}{1996}) points out that \(S\) can be
interpreted as the layer loss (net premium) density. Specifically, \(S\)
is the layer loss density in the sense that
\(S(x)=d/dx(\E[I_{(0, x]}(X)])\) is the marginal rate of increase in
expected losses in the layer at \(x\). We use density in this sense to
define premium, margin and equity densities, in addition to loss
density.

Clearly \(S(x)\) equals the expected loss to the cover \(1_{\{X>x\}}\).
By scaling, \(S(x)dx\) is the close to the expected loss for
\(I_{(x, x+dx]}\) when \(dx\) is very small; Bodoff
(\protect\hyperlink{ref-Bodoff2007}{2007}) calls these infinitesimal
layers.

Wang (\protect\hyperlink{ref-Wang1996}{1996}) goes on to interpret \[
\int_x^{x+y} g(S(t))\,dt
\] as the layer premium and hence \(g(S(x))\) as the layer premium
density. We write \(P(x):=g(S(x))\) for the premium density.

We can decompose \(X\) into a sum of thin layers. All these layers are
comonotonic with one another and with \(X\), resulting in an additive
decomposition of \(\rho(X)\), since \(\rho\) is comonotonic additive.
The decomposition mirrors the definition of \(\rho\) as an integral,
\cref{eq:drm-def}.

The amount of assets \(a\) available to pay claims determines the
quality of insurance, and premium and expected losses are functions of
\(a\). Premiums are well-known to be sensitive to the insurer's asset
resources and solvency, Phillips, Cummins, and Allen
(\protect\hyperlink{ref-Phillips1998}{1998}). Assets may be infinite,
implying unlimited cover. When \(a\) is finite there is usually some
chance of default. Using the layer density view, define expected loss
\(\bar S\) and premium \(\bar P\) functions as \begin{gather}
\bar S(a):=\E[X\wedge a]=\int_0^a S(x)\,dx \label{eq:sbar-def} \\
\bar P(a):=\rho(X\wedge a) = \int_0^\infty g(S_{X\wedge a}(x))\,dx
=\int_0^a g(S_{X}(x))dx. \label{eq:prem-def}
\end{gather}

Margin is \(\bar M(a):=\bar P(a)-\bar S(a)\) and margin density is
\(M(a)=d\bar M(a)/da\). Assets are funded by premium and equity
\(\bar Q(a):=a-\bar P(a)\). Again \(Q(a)=d\bar Q/da = 1-P(a)\). Together
\(S\), \(M\), and \(Q\) give the split of layer funding between expected
loss, margin and equity. Layers up to \(a\) are, by definition, fully
collateralized. Thus \(\rho(X\wedge a)\) is the premium for a
defaultable cover on \(X\) supported by assets \(a\), whereas
\(\rho(X)\) is the premium for an unlimited, default-free cover.

The layer density view is consistent with more standard approaches to
pricing. If \(X\) is a Bernoulli risk with \(\Pr(X=1)=s\) and expected
loss cost \(s\), then \(\rho(X)=g(s)\) can be regarded as pricing a unit
width layer with attachment probability \(s\). In an intermediated
context, the funding constraint requires layers to be fully
collateralized by premium plus equity---without such funding the
insurance would not be credible since the insurer has no other source of
funds.

Given \(g\) we can compute insurance market statistics for each layer.
The loss, premium, margin, and equity densities are \(s\), \(g(s)\),
\(g(s)-s\) and \(1-g(s)\). The layer loss ratio is \(s/g(s)\) and
\((g(s)-s)/(1-g(s))\) is the layer return on equity. These quantities
are illustrated in \cref{fig:glr} for a typical distortion function. The
corresponding statistics for ground-up covers can be computed by
integrating densities.

\begin{figure}
\centering
\includegraphics[width=\textwidth,height=0.3\textheight]{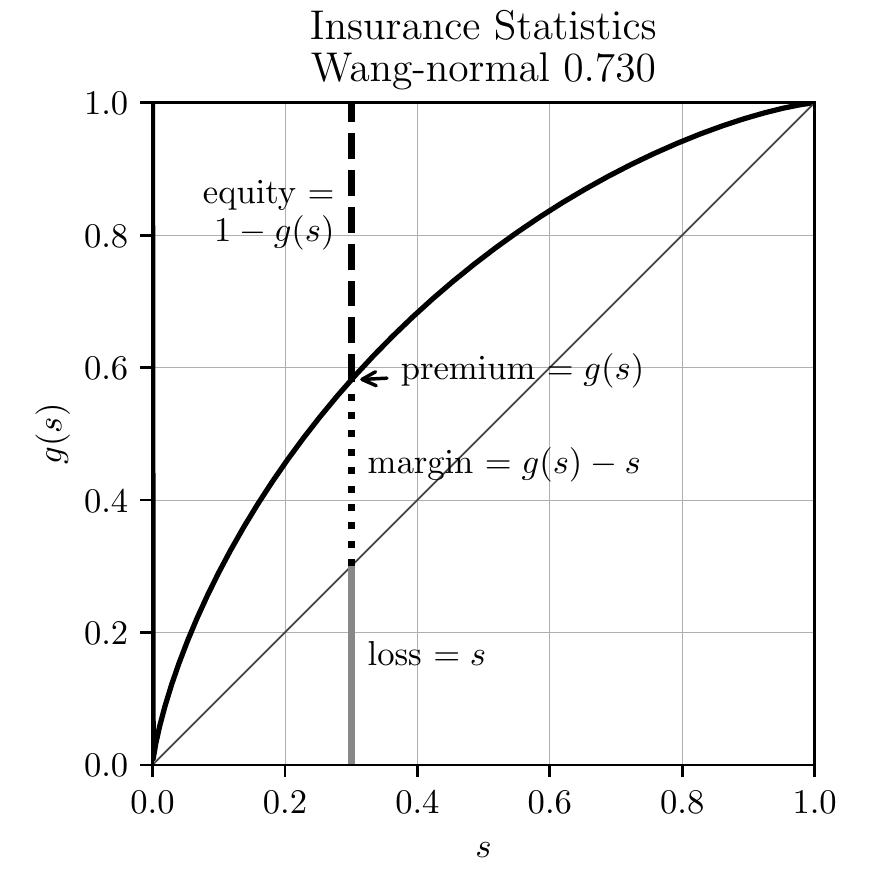}
\caption{Relationship between distortion \(g\) and insurance market
statistics as a function of exceedance probability \(s\).
\label{fig:glr}}
\end{figure}

For an investor-written risk we regard the margin as compensation for
ambiguity aversion and associated winner's curse drag. Both of these
effects are correlated with risk, so the margin is hard to distinguish
from a risk load, but the rationale is different. Again, recall,
although \(\rho\) is non-additive and appears to charge for
diversifiable risk, De Waegenaere, Kast, and Lapied
(\protect\hyperlink{ref-DeWaegenaere2003}{2003}) assures us the pricing
is consistent with a general equilibrium.

The layer density is distinct from models that vary the volume of each
line in a homogeneous portfolio model. Our portfolio is static. By
varying assets we are implicitly varying the quality of insurance.

\hypertarget{the-equal-priority-default-rule}{%
\subsection{The Equal Priority Default
Rule}\label{the-equal-priority-default-rule}}

If assets are finite and the provider has limited liability we need to
to determine policy-level cash flows in default states before we can
determine the fair market value of insurance. The most common way to do
this is using equal priority in default.

Under limited liability, total losses are split between provider
payments and provider default as \[
X = X\wedge a + (X-a)^+.
\] Next, actual payments \(X\wedge a\) must be allocated to each policy.

\(X_i\) is the amount promised to \(i\) by their insurance contract.
Promises are limited by policy provisions but are not limited by
provider assets. At the policy level, equal priority implies the
payments made to, and default borne by, policy \(i\) are split as
\begin{align*}
X_i 
&= X_i \frac{X\wedge a}{X} + X_i \frac{(X-a)^+}{X} \\
&= (\text{payments to policy $i$}) + (\text{default borne by policy $i$}).
\end{align*} Therefore the payments made to policy \(i\) are given by
\begin{equation}
X_i(a) := X_i \frac{X\wedge a}{X}
= \begin{cases}
      X_i  & X \le a \\
      X_i\dfrac{a}{X} & X > a.
  \end{cases}\label{eq:equal-priority}
\end{equation} \(X_i(a)\) is the amount actually paid to policy \(i\).
It depends on \(a\), \(X\) and \(X_i\). The dependence on \(X\) is
critical. It is responsible for almost all the theoretical complexity of
insurance pricing.

It is worth reiterating that with this definition
\(\sum_i X_i(a)=X\wedge a\).

\textbf{Example.} Here is an example illustrating the effect of equal
priority. Consider a certain loss \(X_0=1000\) and \(X_1\) given by a
lognormal with mean 1000 and coefficient of variation 2.0. Prudence
requires losses be backed by assets equal to the 0.9 quantile. On a
stand-alone basis \(X_0\) is backed by \(a_0=1000\) and is risk-free.
\(X_1\) is backed by \(a_1=2272\) and the recovery is subject to a
considerable haircut, since \(\E[X_1\wedge 2272] = 732.3\). If these
risks are pooled, the pool must hold \(a=a_0+a_1\) for the same level of
prudence. When \(X_1\le a_1\) both lines are paid in full. But when
\(X_1 > a_1\), \(X_0\) receives \(1000(a/(1000+X_1))\) and \(X_1\)
receives the remaining \(X_1(a/(1000+X_1))\). Payment to both lines is
pro rated down by the same factor \(a/(1000+X_1)\)---hence the name
\emph{equal} priority. In the pooled case, the expected recovery to
\(X_0\) is 967.5 and 764.8 to \(X_1\). Pooling and equal priority result
in a transfer of 32.5 from \(X_0\) to \(X_1\). This example shows what
can occur when a thin tailed line pools with a thick tailed one under a
weak capital standard with equal priority. We shall see how pricing
compensates for these loss payment transfers, with \(X_1\) paying a
positive margin and \(X_0\) a negative one.

\hypertarget{expected-loss-payments-at-different-asset-levels}{%
\subsection{Expected Loss Payments at Different Asset
Levels}\label{expected-loss-payments-at-different-asset-levels}}

Expected losses paid to policy \(i\) are \(\bar S_i(a) := \E[X_i(a)]\).
\(\bar S_i(a)\) can be computed, conditioning on \(X\), as
\begin{equation}\label{eq:eloss-main}
\bar S_i(a) = \E[\E[X_i(a)\mid X]] = \E[X_i \mid X \le a]F(a) + a\E\left[ \frac{X_i}{X}\mid X>a \right]S(a).
\end{equation} Because of its importance in allocating losses, define
\begin{equation} \label{eq:alpha-def}
\alpha_i(a) := \E[X_i/X\mid X> a].
\end{equation} The value \(\alpha_i(x)\) is the expected proportion of
recoveries by line \(i\) in the layer at \(x\). Since total assets
available to pay losses always equals the layer width, and the chance
the layer attaches is \(S(x)\), it is intuitively clear
\(\alpha_i(x)S(x)\) is the loss density for line \(i\), that is, the
derivative of \(\bar S_i(x)\) with respect to \(x\). We now show this
rigorously.

\medskip

\begin{proposition} Expected losses to policy $i$ under equal priority, when total losses are supported by assets $a$, is given by
\begin{equation} \label{eq:alpha-S}
\bar S_i(a) =\E[X_i(a)] = \int_0^a \alpha_i(x)S(x)dx
\end{equation}
and so the policy loss density at $x$ is $S_i(x):=\alpha_i(x)S(x)$.
\end{proposition}

\begin{proof}
By the definition of conditional expectation, $\alpha_i(a)S(a)=\E[(X_i/X)1_{X>a}]$. Conditioning on $X$, using the tower property, and taking out the functions of $X$ on the right shows
$$
\alpha_i(a)S(a)=\E[\E[(X_i/X) 1_{X>a}\mid X]]=\int_a^\infty \E[X_i \mid X=x]\dfrac{f(x)}{x}dx
$$
and therefore
\begin{equation}\label{eq:alpha-prime}
\frac{d}{da}(\alpha_i(a)S(a)) = -\E[X_i \mid X=a]\dfrac{f(a)}{a}.
\end{equation}
Now we can use integration by parts to compute
\begin{align*}
\int_0^a \alpha_i(x)S(x)\,dx
&= x\alpha_i(x)S(x)\Big\vert_0^a + \int_0^a x\,\E[X_i \mid X=x]\dfrac{f(x)}{x}\,dx\\
&= a\alpha_i(a)S(a) + E[X_i \mid X\le a]F(a) \\
&=  \bar S_i(a)
\end{align*}
by \cref{eq:eloss-main}. Therefore the policy $i$ loss density in the asset layer at $a$, i.e. the derivative of \cref{eq:eloss-main} with respect to $a$, is $S_{i}(a)=\alpha_i(a) S(a)$ as required.
\end{proof}

To recap, \cref{eq:alpha-S} gives a direct analog to \cref{eq:sbar-def}
for policy \(i\) losses. Note that \(S_i\) is \emph{not} the survival
function of \(X_i(a)\) nor of \(X_i\). \Cref{eq:alpha-S} is surprising
because it gives a decomposition of \(S\) \emph{through} the convolution
of random variables: \[
\begin{tikzcd}[row sep=scriptsize, column sep=scriptsize]
X       \arrow[r, "="]                 &  \sum_i X_i                           & & X\wedge a     \arrow[r, "="]                & \sum_i X_i(a)                              \\
\E[X]   \arrow[r, "="] \arrow[d, "="]  &  \sum_i \E[X_i]      \arrow[d, "="]   & & \E[X\wedge a] \arrow[r, "="] \arrow[d, "="] & \sum_i \E[X_i(a)]          \arrow[d, "="]  \\
\dint S \arrow[r, "="]                 &  \sum_i \dint S_{X_i}                 & & \dint_0^a S   \arrow[r, "="]                & \sum_i\dint_0^a \alpha_i S                 \\
S \arrow[r, "\not="]                   &  \sum S_{X_i},                        & & S \arrow[r, "="]                            &  \sum \alpha_i S.                          \\
\end{tikzcd}
\]

\hypertarget{premiums-at-different-asset-levels}{%
\subsection{Premiums at Different Asset
Levels}\label{premiums-at-different-asset-levels}}

Premium under \(\rho\) is given by \cref{eq:prem-def}. We can interpret
\(g(S(a))\) as the portfolio premium density in the layer at \(a\). We
now consider the premium and premium density for each policy.

Using integration by parts we can express the price of an unlimited
cover on \(X\) as \begin{equation}
\label{eq:nat1}
\rho(X)=\int_0^\infty g(S(x))\,dx = \int_0^\infty xg'(S(x))f(x)\,dx = \E[Xg'(S(X)))].
\end{equation} It is important that this integral is over all \(x\ge 0\)
so the \(xg(S(x))\vert_0^a\) term disappears. \Cref{eq:nat1} makes sense
because a concave distortion is continuous on \((0,1]\) and can have at
most countably infinitely many points where it is not differentiable (it
has a kink). In total these points have measure zero, Borwein and
Vanderwerff (\protect\hyperlink{ref-Borwein2010}{2010}), and we can
ignore them in the integral. For more details see Dhaene et al.
(\protect\hyperlink{ref-Dhaene2012b}{2012}).

By \Cref{eq:nat1}, and the properties of a distortion function,
\(g'(S(X))\) is the Radon-Nikodym derivative of a measure \(\mathsf Q\)
with \(\rho(X)=\E_{\mathsf Q}[X]\). In fact,
\(\E_{\mathsf Q}[Y]=\E[Yg'(S(X))]\) for all random variables \(Y\). In
general, any non-negative function \(Z\) (measure \(\mathsf Q\)) with
\(\E[Z]=1\) and \(\rho(X)=\E[XZ]\) (\(=\E_{\mathsf Q}[X]\)) is called a
contact function (subgradient) for \(\rho\) at \(X\), see Shapiro,
Dentcheva, and Ruszczyński (\protect\hyperlink{ref-Shapiro2009}{2009}).
Thus \(g'(S(X))\) is a contact function for \(\rho\) at \(X\). The name
subgradient comes from the fact that
\(\rho(X+Y)\ge \E_{\mathsf Q}[X+Y] = \rho(X) + \E_{\mathsf Q}[Y]\), by
\cref{thm:coherent}. The set of subgradients is called the
subdifferential of \(\rho\) at \(X\). If there is a unique subgradient
then \(\rho\) is differentiable. Delbaen
(\protect\hyperlink{ref-Delbaen2000}{2000}) Theorem 17 shows that
subgradients are contact functions.

We can interpret \(g'(S(X))\) as a state price density specific to the
\(X\), suggesting that \(\E[X_ig'(S(X))]\) gives the value of the cash
flows to policy \(i\). This motivates the following definition.

\medskip

\begin{definition} For $X=\sum_i X_i$ with $\mathsf Q\in\mathcal Q$ so that $\rho(X)=\E_{\mathsf Q}[X]$, the {\bf natural allocation premium} to policy $X_j$ as part of the portfolio $X$ is $\E_{\mathsf Q}[X_j]$. It is denoted $\rho_X(X_j)$.
\end{definition}

The natural allocation premium is a standard approach, appearing in
Delbaen (\protect\hyperlink{ref-Delbaen2000}{2000}), Venter, Major, and
Kreps (\protect\hyperlink{ref-Venter2006}{2006}) and Tsanakas and
Barnett (\protect\hyperlink{ref-Tsanakas2003a}{2003}) for example. It
has many desirable properties. Delbaen shows it is a fair allocation in
the sense of fuzzy games and that it has a directional derivative,
marginal interpretation when \(\rho\) is differentiable. It is
consistent with Jouini and Kallal
(\protect\hyperlink{ref-Jouini2001}{2001}) and Campi, Jouini, and Porte
(\protect\hyperlink{ref-Campi2013}{2013}), which show the rational price
of \(X\) in a market with frictions must be computed by state prices
that are anti-comonotonic \(X\). In our application the signs are
reversed: \(g'(S(X))\) and \(X\) are comonotonic.

The choice \(g'(S(X))\) is economically meaningful because it weights
the largest outcomes of \(X\) the most, which is appropriate from a
social, regulatory and investor perspective. It is also the only choice
of weights that works for all levels of assets. Since investors stand
ready to write any layer at the price determined by \(g\), their
solution must work for all \(a\).

However, there are two technical issues with the proposed natural
allocation. First, unlike prior works, we are allocating the premium for
\(X\wedge a\), not \(X\), a problem also considered in Major
(\protect\hyperlink{ref-Major2018}{2018}). And second, \(\mathsf Q\) may
not be unique. In general, uniqueness fails at capped variables like
\(X\wedge a\). Both issues are surmountable for a DRM, resulting in a
unique, well defined natural allocation. For a non-comonotonic additive
risk measure this is not the case.

It is helpful to define the premium, risk adjusted, analog of the
\(\alpha_i\) as \begin{equation} \label{eq:beta-def}
\beta_i(a) := \E_{\mathsf Q}[(X_i/X) \mid X > a].
\end{equation} \(\beta_i(x)\) is the value of the recoveries paid to
line \(i\) by a policy paying 1 in states \(\{ X>a \}\), i.e.~an
allocation of the premium for \(1_{X>a}\). By the properties of
conditional expectations, we have \begin{equation}
\label{eq:beta-cond}
\beta_i(a) = \frac{\E[(X_i/X) Z\mid X > a]}{\E[Z\mid X > a]}.
\end{equation} The denominator equals \(\mathsf Q(X>a)/\mathsf P(X>a)\).
Remember that while \(\E_{\mathsf Q}[X]=\E[XZ]\), for conditional
expectations
\(\E_{\mathsf Q}[X\mid \mathcal F]=\E[XZ\mid \mathcal F]/\E[Z\mid \mathcal F]\),
see Föllmer and Schied (\protect\hyperlink{ref-Follmer2011}{2011}),
Proposition A.12.

To compute \(\alpha_i\) and \(\beta_i\) we use a third function,
\begin{equation}\label{eq:kappa-def}
\kappa_i(x):= \E[X_i \mid X=x],
\end{equation} the conditional expectation of loss by policy, given the
total loss. It is an important fact that the risk adjusted version of
\(\kappa\) is unchanged because DRMs are law invariant. With these
preliminaries we can state the main theorem of this section.

\medskip

\begin{theorem} \label{main-theorem} Let $\mathsf Q\in \mathcal Q$ be the measure with Radon-Nikodym derivative $Z=g'(S_X(X))$.
\begin{enumerate}
\item $\E[X_i \mid X=x]=\E_{\mathsf Q}[X_i \mid X=x]$.
\item $\beta_i$ can be computed from $\kappa_i$ as
\begin{equation}
\beta_i(a)= \frac{1}{\mathsf Q(X>a)}\int_a^\infty \dfrac{\kappa_i(x)}{x} g'(S(x))f(x)\, dx. \label{eq:beta-easy}
\end{equation}
\item The natural allocation premium for policy $i$ under equal priority when total losses are supported by assets $a$, $\bar P_i(a):=\rho_{X\wedge a}(X_i(a))$, is given by
\begin{align}
\bar P_i(a) &=
\E_{\mathsf Q}[X_i \mid {X\le a}](1-g(S(a))) + a\E_{\mathsf Q}[X_i/X  \mid {X > a}]g(S(a)) \label{eq:pibar-main} \\
&=\E[X_iZ\mid X\le a](1-S(a)) + a\E[(X_i/X)Z\mid X>a]S(a).
\end{align}
\item The policy $i$ premium density is
\begin{equation}
P_i(a)=\beta_i(a)g(S(a)).
\label{eq:beta-gS}
\end{equation}
\end{enumerate}
\end{theorem}

The Theorem offers two contributions. First, it shows we can replace
\(\E_{\mathsf Q}[X_i \mid X]\) with \(\E[X_i \mid X]\), which enables
explicit calculation. There is no risk adjusted version of \(\kappa_i\).
Intuitively, a law invariant risk measure cannot change probabilities
within an event defined by \(X\): if it did then it would be
distinguishing between events on information other than \(S(X)\) whereas
law invariance says this is all that can matter. And second, it
identifies the premium density \cref{eq:beta-gS}, giving an allocation
of \cref{eq:prem-def} and a premium analog of \cref{eq:alpha-S}. It
provides a clear and illuminating way to visualize risk by collapsing a
multidimensional problem to one dimension, see \cref{fig:eg-2}. Part (4)
provides a direct premium analog of \cref{eq:eloss-main} and part (5) an
analog of \cref{eq:alpha-S}. \Cref{eq:pibar-main} is the same as
Tsanakas and Barnett (\protect\hyperlink{ref-Tsanakas2003a}{2003}) eq.
(19), although their derivation is in the context of a homogeneous
portfolio whereas our portfolio is static.

\noindent{\it Proof.\ \ } Part (1) follows in the same way as
\cref{eq:beta-cond}.

Since \(\rho\) is comonotonic additive
\(\rho(X)=\rho(X\wedge a) + \rho((X-a)^+)\) and hence
\(\rho(X)=\E_{\mathsf{Q}}[X\wedge a] + \E_{\mathsf{Q}}[(X-a)^+] \le \rho(X\wedge a) + \rho((X-a)^+)=\rho(X)\).
But since \(\E_{\mathsf{Q}}[X\wedge a] \le \rho(X\wedge a)\) and
\(\E_{\mathsf{Q}}[(X-a)^+] \le \rho((X-a)^+)\) it follows
\(\E_{\mathsf{Q}}[X\wedge a] = \rho(X\wedge a)\) and
\(\E_{\mathsf{Q}}[(X - a)^+] = \rho((X - a)^+)\). Therefore the contact
functions for \(X\) and \(X\wedge a\) are the same and it is legitimate
to assume \(Z=g'(S(X))\) when allocating premium for \(X\wedge a\).

To prove Part (2), note that by \cref{eq:beta-cond}
\(\beta_i(a)g(S(a))=\E_{\mathsf{Q}}[(X_i/X) 1_{X>a}]\). Conditioning on
\(X\), using the tower property, and taking out the known functions of
\(X\) on the right, shows \begin{align*}
\beta_i(a)g(S(a))
&= \E[\E[(X_i/X) g'(S(X)) 1_{X>a}\mid X]] \\
&= \E[(\E[(X_i\mid X]/X) g'(S(X)) 1_{X>a} ]] \\
&= \int_a^\infty \dfrac{\E[X_i \mid X=x]}{x} g'(S(x))f(x)\,dx.
\end{align*}

It follows from the definition of \(X_i(a)\), \cref{eq:equal-priority},
and the fact \(Z\) is a contact function for \(X\wedge a\) that
\begin{align*}
\bar P_i(a)
&= \E[X_i(a) g'(S(X))]  \\
&= \E[X_i g'(S(X)) 1_{X\le a}] +  \E[a(X_i/X) g'(S(X)) 1_{X > a}]  \\
&= \E[X_i g'(S(X)) \mid {X\le a}](1-S(a)) +  \\
&\qquad\qquad a\E[(X_i/X) g'(S(X)) \mid {X > a}]S(a) \\
&= \E_{\mathsf Q}[X_i \mid {X\le a}](1-g(S(a))) +  \\
&\qquad\qquad a\E_{\mathsf Q}[(X_i/X) \mid {X > a}]g(S(a)),
\end{align*} giving Part (3).

Rearranging \cref{eq:beta-easy} and differentiating gives
\begin{equation*}
\frac{d}{da}(\beta_i(a)g(S(a))) = -\dfrac{\E[X_i \mid X=a]}{a} g'(S(a))f(a).
\end{equation*} Now use integration by parts to compute \begin{align*}
\int_0^a \beta_i(x)g(S(x))\,dx
&= x\beta_i(x)g(S(x))\Big\vert_0^a + \int_0^a x\,\dfrac{\E[X_i \mid X=x]}{x}  g'(S(x))f(x) \,dx\\
&= a\beta_i(a)g(S(a)) + \E_{\mathsf{Q}}[X_i \mid X\le a](1-g(S(a)) \\
&=  \bar P_i(a)
\end{align*} by \cref{eq:pibar-main}. As a result, the policy \(i\)
premium density in the asset layer at \(a\), i.e.~the derivative of
\(\bar P_{i}(a)\) with respect to \(a\), is
\(P_{i}(a) =\beta_i(a) g(S(a))\), giving Part (4).

\qed

The proof writes the price of a limited liability cover as the price of
default-free protection minus the value of the default put. This is the
standard starting point for allocation in a perfect competitive market
taken by Phillips, Cummins, and Allen
(\protect\hyperlink{ref-Phillips1998}{1998}), Myers and Read Jr.
(\protect\hyperlink{ref-Myers2001}{2001}), Sherris
(\protect\hyperlink{ref-Sherris2006a}{2006}), and Ibragimov, Jaffee, and
Walden (\protect\hyperlink{ref-Ibragimov2010}{2010}). They then allocate
the default put rather than the value of insurance payments directly.

The problem that can occur when \(\mathsf Q\) is not unique, but that
can be circumvented when \(\rho\) is a DRM, can be illustrated as
follows. Suppose \(\rho\) is given by \(p\)-TVaR. The measure
\(\mathsf{Q}\) weights the worst \(1-p\) proportion of outcomes of \(X\)
by a factor of \((1-p)^{-1}\) and ignores the others. Suppose \(a\) is
chosen as \(p'\)-VaR for a lower threshold \(p'<p\). Let
\(X_a=X\wedge a\) be capped insured losses and \(C=\{X_a=a\}\). By
definition \(\Pr(C)\ge 1-p'>1-p\). Pick any \(A\subset C\) of measure
\(1-p\) so that \(\rho(X)=\E[X\mid A]\). Let \(\psi\) be a measure
preserving transformation of \(\Omega\) that acts non-trivially on \(C\)
but trivially off \(C\). Then \(\mathsf{Q}'=\mathsf Q\psi\) will satisfy
\(\E_{\mathsf{Q}'}[X_a]=\E_{\mathsf{Q}}[X_a\psi^{-1}]=\rho(X_a)\) but in
general \(\E_{\mathsf{Q}'}[X]<\rho(X)\). The natural allocation with
respect to \(\mathsf{Q}'\) will be different from that for
\(\mathsf{Q}\). The theorem isolates a specific \(\mathsf Q\) to obtain
a unique answer. The same idea applies to \(\mathsf Q\) from other,
non-TVaR, \(\rho\): you can always shuffle part of the contact function
within \(C\) to generate non-unique allocations. See \cref{example-1}
for an example.

To recap: the premium formulas \cref{eq:pibar-main,eq:beta-gS} have been
derived assuming capital is provided at a cost \(g\) and there is equal
priority by line. They are computationally tractable and require no
other assumptions. There is no need to assume the \(X_i\) are
independent. They produce an entirely general, canonical determination
of premium in the presence of shared costly capital. This result extends
Grundl and Schmeiser (\protect\hyperlink{ref-Grundl2007}{2007}), who
pointed out that with an additive pricing functional there is no need to
allocate capital in order to price, to the situation of a non-additive
DRM pricing functional.

The key formulas we have derived are summarized in
\cref{tab:natural-allocation-summary}.

\begin{sidewaystable}
\begin{tikzpicture}[
  auto,
  transform shape,
    table/.style={matrix of nodes,
      row sep=3pt,    
      column sep=3pt, 
      nodes in empty cells,
      nodes={draw=white, rectangle, scale=0.85, text badly ragged},  
      column 1/.style={below right, , align=left, text width=3cm,  nodes={text=black, font=\bfseries}},       
      column 2/.style={below right, align=left, text width=10cm},                                             
      column 3/.style={below right, align=left, text width=13cm},                                             
      row 1/.style={nodes={text=black, font=\bfseries}, minimum height=25pt},  
    }]

\matrix (natural) [table, ampersand replacement=\&]{
Quantity                      \& Loss                                                             \& Premium                                                                     \\
Cash flow                     \& $X_i(a)=X_i\dfrac{X\wedge a}{X}$                                 \& N/a                                                                         \\
                              \&                                                                  \&                                                                             \\
Measure                       \& Objective, $S(x)$, $f(x)$                                        \& Risk adjusted, $\mathsf Q$, $g(S(x))$, $g'(S(x))f(x)$                             \\
                              \&                                                                  \&                                                                             \\
Expectation                   \& $\bar S_i(a)=\E[X_i(a)]$                                         \& $\bar P_i(a)=\E_{\mathsf Q}[X_i(a)]=\E[X_i(a)g'(S(X))]$                           \\
                              \&                                                                  \&                                                                             \\
Conditioning expectation      \& $\E[X_i\mid X\le a]F(a) + a\E[X_i/X\mid X >a]S(a)$               \& $\E_{\mathsf Q}[X_i\mid X\le a](1-g(S(a))) + a\E_{\mathsf Q}[X_i/X\mid X >a]g(S(a))$    \\
                              \&                                                                  \&                                                                             \\
Share function                \& $\alpha_i(x) =\E[X_i/X\mid X>x]$                                 \& $\beta_i(x) =\E_{\mathsf Q}[X_i/X\mid X>x]$                       \\
                              \&                                                                  \&                                                                             \\
Derivative of share function  \& $(\alpha_i S)'(x)=-\E[X_i\mid X=x]f(x)/x=-\kappa_i(x)f(x) / x$   \& $(\beta_i g(S))'(x)=-\E[X_i\mid X=x]g'(S(x))f(x)/x=-\kappa_i(x)g'(S(x))f(x) / x$  \\
                              \&                                                                  \&                                                                             \\
Lee integral expectation      \& $\dint_0^a \alpha_i(x)S(x)\,dx$                                  \& $\dint_0^a \beta_i(x)g(S(x))\,dx$                                           \\
                              \&                                                                  \&                                                                             \\
Outcome integral expectation  \& $\dint_0^a \kappa_i(x) f(x)\,dx + a\alpha_i(a)S(a)$              \& $\dint_0^a \kappa_i(x)g'(S(x))f(x)\,dx + a\beta_i(a)g(S(a))$            \\
                              \&                                                                  \&                                                                             \\
Scenario integral expectation \& $\dint_0^{F(a)} \kappa_i(q(p))\,dp + a\alpha_i(a)S(a)$           \& $\dint_0^{1-g(S(a))} \kappa_i(q(1-g^{-1}(1-p)))\,dp + a\beta_i(a)g(S(a))$     \\
                              \&                                                                  \&                                                                              \\
};
\path[draw, thick]     (natural-1-1.north west)  -- (natural-1-3.north east);
\path[draw, very thin] (natural-1-1.south west)  -- (natural-1-3.south east);
\path[draw, very thin] (natural-3-1.south west)  -- (natural-3-3.south east);
\path[draw, very thin] (natural-5-1.south west)  -- (natural-5-3.south east);
\path[draw, very thin] (natural-7-1.south west)  -- (natural-7-3.south east);
\path[draw, very thin] (natural-9-1.south west)  -- (natural-9-3.south east);
\path[draw, very thin] (natural-11-1.south west) -- (natural-11-3.south east);
\path[draw, very thin] (natural-13-1.south west) -- (natural-13-3.south east);
\path[draw, very thin] (natural-15-1.south west) -- (natural-15-3.south east);
\path[draw, very thin] (natural-17-1.south west) -- (natural-17-3.south east);
\path[draw, thick]     (natural-19-1.south west) -- (natural-19-3.south east);

\end{tikzpicture}
\caption{Different ways of computing expected losses and the natural allocation. \label{tab:natural-allocation-summary} }
\end{sidewaystable}

\hypertarget{properties-of-alpha-beta-and-kappa}{%
\section{Properties of Alpha, Beta, and
Kappa}\label{properties-of-alpha-beta-and-kappa}}

In this section we explore properties of \(\alpha_i\), \(\beta_i\), and
\(\kappa_i\), see \cref{eq:alpha-def}, \cref{eq:beta-def}, and
\cref{eq:kappa-def}, and show how they interact to determine premiums by
line via the natural allocation.

For a measurable \(h\), \(\E[X_ih(X)]=\E[\kappa_i(X)h(X)]\) by the tower
property. This simple observation results in huge simplifications. In
general, \(\E[X_ih(X)]\) requires knowing the full bivariate
distribution of \(X_i\) and \(X\). Using \(\kappa_i\) reduces it to a
one dimensional problem. This is true even if the \(X_i\) are
correlated. The \(\kappa_i\) functions can be estimated from data using
regression and they provide an alternative way to model correlations.

Despite their central role, the \(\kappa_i\) functions are probably
unfamiliar so we begin by giving several examples to illustrate how they
behave. In general, they are non-linear and usually, but not always,
increasing.

\hypertarget{examples-of-kappa-functions}{%
\subsection{\texorpdfstring{Examples of \(\kappa\)
functions}{Examples of \textbackslash kappa functions}}\label{examples-of-kappa-functions}}

\begin{enumerate}
\def\labelenumi{\arabic{enumi}.}
\item
  If \(Y_i\) are independent and identically distributed and
  \(X_n=Y_1+\cdots +Y_n\) then \(\E[X_m\mid X_{m+n}=x]=mx/(m+n)\) for
  \(m\ge 1, n\ge 0\). This is obvious when \(m=1\) because the functions
  \(\E[Y_i\mid X_n]\) are independent across \(i=1,\ldots,n\) and sum to
  \(x\). The result follows because conditional expectations are linear.
  In this case \(\kappa_i(x)=mx/(m+n)\) is a line through the origin.
\item
  If \(X_i\) are multivariate normal then \(\kappa_i\) are straight
  lines, given by the usual least-squares fits \[
  \kappa_i(x)= \E[X_i] + \frac{\cov(X_i,X)}{\var(X)}(x-\E[X]).
  \] This example is familiar from the securities market line and the
  CAPM analysis of stock returns. If \(X_i\) are iid it reduces to the
  previous example because the slope is \(1/n\).
\item
  If \(X_i\), \(i=1,2\), are compound Poisson with the same severity
  distribution then \(\kappa_i\) are again lines through the origin.
  Suppose \(X_i\) has expected claim count \(\lambda_i\). Write the
  conditional expectation as an integral, expand the density of the
  compound Poisson by conditioning on the claim count, and then swap the
  sum and integral to see that
  \(\kappa_1(x)=\E[X_1\mid X_1 + X_2=x]=x\,\E[N(\lambda_1)/(N(\lambda_1)+N(\lambda_2))]\)
  where \(N(\lambda)\) are independent Poisson with mean \(\lambda\).
  This example generalizes the iid case. Further conditioning on a
  common mixing variable extends the result to mixed Poisson frequencies
  where each aggregate can have a separate or shared mixing
  distribution. The common severity is essential. The result means that
  if a line of business is defined to be a group of policies that shares
  the same severity distribution, then premiums for policies within the
  line will have rates proportional to their expected claim counts.
\item
  A theorem of Efron says that if \(X_i\) are independent and have
  log-concave densities then all \(\kappa_i\) are non-decreasing,
  Saumard and Wellner (\protect\hyperlink{ref-Saumard2014}{2014}). The
  multivariate normal example is a special case of Efron's theorem.
\end{enumerate}

Denuit and Dhaene (\protect\hyperlink{ref-Denuit2012}{2012}) define an
ex post risk sharing rule called the conditional mean risk allocation by
taking \(\kappa_i(x)\) to be the allocation to policy \(i\) when
\(X=x\). A series of recent papers, see Denuit and Robert
(\protect\hyperlink{ref-Denuit2020e}{2020}) and references therein,
considers the properties of the conditional mean risk allocation
focusing on its use in peer-to-peer insurance and the case when
\(\kappa_i(x)\) is linear in \(x\).

\hypertarget{the-behavior-of-alpha_i-beta_i}{%
\subsection{\texorpdfstring{The Behavior of \(\alpha_i\),
\(\beta_i\)}{The Behavior of \textbackslash alpha\_i, \textbackslash beta\_i}}\label{the-behavior-of-alpha_i-beta_i}}

By definition \(\alpha_i(x)\) is the expected proportion of losses from
policy \(i\) in \(1_{\{X>x\}}\) and \(\beta_i(x)\) is the risk adjusted
proportion. They are average proportions not proportions of the
averages:
\(\alpha_i(x) = \E[X_i /X \mid X> x]\not=\E[X_i\mid X> x]/\E[X\mid X>x]\)
because of Jensen's inequality applied to the convex function
\(x\mapsto 1/x\).

To better understand the shape of \(\alpha_i\) and \(\beta_i\) we can
compute their derivatives. Differentiating
\(\alpha_i(x)S(x)=\E[(X_i/X)1_{X>x}]\) and re-arranging gives
\begin{equation}  \label{eq:alpha-prime}
\alpha_i'(x) = \left( \alpha_i(x) - \frac{\kappa_i(x)}{x} \right)\frac{f(x)}{S(x)}.
\end{equation} The results for \(\beta_i\) are analogous. The function
\(h(x):=f(x)/S(x)\) is called the hazard rate. If \(X\) models a
lifetime, \(h\) is called the force of mortality. For thick right-skewed
distributions \(h\) is typically an eventually decreasing function. For
thin tailed distributions it is typically an eventually increasing
function. It is constant for the exponential distribution. The action of
\(g\) is to make the right tail thicker and so to decrease the hazard
rate. Since \(h(x)=-d/dx(\log(S(x)))\) it follows that \[
S(x) = \exp\left(-\int_t^\infty h(s)ds\right).
\] The integral is called the cumulative hazard function. From this
formulation it is clear the proportional hazard \(g(s)=s^r\),
\(0<r\le 1\), acts on the hazard function as multiplication by \(r\),
hence justifying its name.

\Cref{eq:alpha-prime} shows that \(\alpha_i'(x)=0\) if \(f(x)=0\) and
\(S(x)\) close to 1, which will occur in the extreme left tail when
\(X\) includes some level of near certain losses. Then \(\alpha_i\) will
be flat for small \(x\), while \(f(x)\approx 0\). Flat behavior can also
occur if \(\alpha_i(x)-\kappa_i(x)/x=0\), but that is an exceptional
circumstance.

For thick tailed insurance distributions \(h(x)\) is eventually
decreasing but remains strictly positive. If \(\kappa_i(x)/x\) is
decreasing then \(\alpha_i'(x)<0\) because \(\alpha_i(x)\) is the
probability weighted integral of \(\kappa_i(t)/t\) over \(t>x\), and so
\(\alpha_i(x)<\kappa_i(x)/x\). Conversely if \(\kappa_i(x)/x\) is
increasing \(\alpha'_i(x)>0\).

Since \(\sum_i \kappa_i(x)=x\) it follows that
\(\sum_i \kappa_i'(x)=1\). It is typical for the thickest tail
distribution, \(i\) say, to behave like
\(\kappa_i(x)\approx x -\sum_{j\not=i} \E[X_j]\) for large \(x\). Then
\(\kappa_i'(x)=1\) and the remaining \(\kappa_j(x)\approx \E[X_j]\) are
almost constant for large \(x\). In that case
\(\kappa_j(x)/x > \alpha_j(x)\) and so \(\alpha_j'(x)<0\) and
\(\alpha_i'(x)>0\). To have two policies with \(\alpha_i\) increasing
requires a very delicate balancing of the thickness of their tails with
\(\kappa_i(x)\), growing with order \(x\). A compound Poisson with the
same severity is an example.

\hypertarget{properties-of-the-natural-allocation}{%
\section{Properties of the Natural
Allocation}\label{properties-of-the-natural-allocation}}

\hypertarget{aggregate-properties}{%
\subsection{Aggregate Properties}\label{aggregate-properties}}

We now explore margin, equity, and return in total and by policy. We
begin by considering them in total.

By definition the average return with assets \(a\) is
\begin{equation}\label{eq:avg-roe}
\bar\iota(a) := \frac{\bar M(a)}{\bar Q(a)}
\end{equation} where margin \(\bar M\) and equity \(\bar Q\) are defined
in the paragraph following \cref{eq:prem-def}.

\Cref{eq:avg-roe} has important implications. It tells us the investor
priced expected return varies with the level of assets. For most
distortions return decreases with increasing capital. In contrast, the
standard RAROC models use a fixed average cost of capital, regardless of
the overall asset level, Tasche
(\protect\hyperlink{ref-Tasche1999}{1999}). CAPM or the Fama-French
three factor model are often used to estimate the average return, with a
typical range of 7 to 20 percent, Cummins and Phillips
(\protect\hyperlink{ref-Cummins2005}{2005}). A common question of
working actuaries performing capital allocation is about so-called
excess capital, if the balance sheet contains more capital than is
required by regulators, rating agencies, or managerial prudence. Our
model suggests that higher layers of capital are cheaper, but not free,
addressing this concern.

The varying returns in \cref{eq:avg-roe} may seem inconsistent with
Miller Modigliani. But that says the cost of funding a given amount of
capital is independent of how it is split between debt and equity; it
does not say the average cost is constant as the amount of capital
varies.

\hypertarget{no-undercut-and-positive-margin-for-independent-risks}{%
\subsection{No-Undercut and Positive Margin for Independent
Risks}\label{no-undercut-and-positive-margin-for-independent-risks}}

The natural allocation has two desirable properties. It is always less
than the stand-alone premium, meaning it satisfies the no-undercut
condition of Denault (\protect\hyperlink{ref-Denault2001}{2001}), and it
produces non-negative margins for independent risks.

\medskip

\begin{proposition}\label{prop:pos-margin}
Let $X=\sum_{i=1}^n X_i$, $X_i$ non-negative and independent, and let $g$ be a distortion. Then (1) the natural allocation is never greater than the stand-alone premium and (2) the natural allocation to every $X_i$ contains a non-negative margin.
\end{proposition}

\noindent{\it Proof.\ \ } It is enough to prove for \(n=2\) by
considering \(X_1\) and \(X_2' = X_2+\cdots +X_n\).

By \cref{thm:coherent} we know that \(\rho(X)=\E_{\mathsf{Q}}[X]\) where
\({\mathsf{Q}}\) has Radon-Nikodym derivative \(g'(S_{X}(X))\). By
definition, the natural allocation is \(\bar P_1=\E[X_1g'(S_X(X))]\).
Therefore, \[
\bar P_1 = \E[X_1g'(S_X(X))] \le \sup_{Q\in\mathcal Q}\E_Q[X] = \rho(X_1)
\] which shows Part (1), that the natural allocation is never greater
than the stand-alone premium.

Let \(\tilde X_1 = X_1 + \E[X_2]\) and \(\tilde X_2 = X_2 - \E[X_2]\).
Then by Rothschild and Stiglitz
(\protect\hyperlink{ref-Rothschild1970}{1970}) or Machina and Pratt
(\protect\hyperlink{ref-Machina1997}{1997})
\(\tilde X_1 + \tilde X_2 \succeq^2 \tilde X_1\), where \(\succeq^2\)
denotes second order stochastic dominance. Bäuerle and Müller
(\protect\hyperlink{ref-Bauerle2006}{2006}) shows that DRMs respect
second order stochastic dominance. Therefore \[
\rho(\tilde X_1 + \tilde X_2)\ge \rho(\tilde X_1).
\] By translation invariance \(\rho(\tilde X_1)=\rho(X_1) + \E[X_2]\).
Since \(\tilde X_1 + \tilde X_2 = X_1 + X_2\) we conclude \[
\rho(X_1 + X_2)\ge \rho(X_1) + \E[X_2].
\]

Combining these results we get \begin{align*}
\bar P_1 + \bar P_2 = \rho(X_1+X_2) &\ge \rho(X_1) + \E[X_2] \\
\implies \bar P_2 &\ge \rho(X_1)-\bar P_1 + \E[X_2] \\
&\ge \E[X_2]
\end{align*} as required for Part (2). \qed

Part (1) is well known. The proof if Part (2) leverages the fact
\(\rho\) is translation invariant. If we add \(X_2=c\) to \(X_1\) then
its natural allocation is \(c\). In a sense, this is the \emph{best
case}. Any non-constant independent variable, no matter how low its
variance, must slightly increase risk. It does not make sense to grant
the new variable a credit off expected loss, when we would not do so for
a constant. A credit is possible for dependent variables, however.

Since \(\bar P_i = \E[\kappa_i(X)g'(S(X))]\) we see the no-undercut
condition holds if \(\kappa_i(X)\) and \(g'(S(X))\) are comonotonic, and
hence if \(\kappa_i\) is increasing, or if \(\kappa_i(X)\) and \(X\) are
positively correlated (recall \(\E[g'(S(X))]=1\)). Since
\(\sum_i \kappa_i(x)=x\) at least one \(\kappa_i\), say
\(\kappa_{i^*}\), must be increasing. Policy \(i^*\) is the capacity
consuming line that will always have a positive margin. In this way
\(\kappa\) differentiates \emph{relative} tail thickness.

\hypertarget{policy-level-properties-varying-with-asset-level}{%
\subsection{Policy Level Properties, Varying with Asset
Level}\label{policy-level-properties-varying-with-asset-level}}

We start with a corollary of the results in
\cref{loss-and-premium-by-policy-and-layer} which gives a nicely
symmetric and computationally tractable expression for the natural
margin allocation in the case of finite assets.

\begin{corollary}
The margin density for line $i$ at asset level $a$ is given by
\begin{equation}\label{eq:coc-by-line}
M_i(a) =\beta_i(a)g(S(a)) -  \alpha_i(a)S(a).
\end{equation}
\end{corollary}

\noindent{\it Proof.\ \ } Using \cref{eq:alpha-S,eq:beta-gS} we can
compute margin \(\bar M_i(a)\) in \(\bar P_i(a)\) by line as
\begin{align}
\bar M_i(a)=& \bar P_i(a) - \bar L_i(a) \nonumber \\
=& \int_0^a \beta_i(x)g(S(x)) -  \alpha_i(x)S(x)\,dx.  \label{eq:margin-by-line}
\end{align} Differentiating we get the margin density for line \(i\) at
\(a\) expressed in terms of \(\alpha_i\) and \(\beta_i\) as shown. \qed

Margin in the current context is the cost of capital, thus
\cref{eq:coc-by-line} is an important result. It allows us to compute
economic value by line and to assess static portfolio performance by
line---one of the motivations for performing capital allocation in the
first place. In many ways it is also a good place to stop. Remember
these results only assume we are using a distortion risk measure and
have equal priority in default. We are in a static model, so questions
of portfolio homogeneity are irrelevant. We are not assuming \(X_i\) are
independent.

What does \cref{eq:coc-by-line} say about by margins by line? Since
\(g\) is increasing and concave \(P(a)=g(S(a))\ge S(a)\) for all
\(a\ge 0\). Thus all asset layers contain a non-negative total margin
density. It is a different situation by line, where we can see \[
M_i(a) \ge 0 \iff
\beta_i(a)g(S(a)) - \alpha_i(a)S(a)\ge 0 \iff
\frac{\beta_i(a)}{\alpha_i(a)} \ge \frac{S(a)}{g(S(a))}.
\] The line layer margin density is positive when \(\beta_i/\alpha_i\)
is greater than the all-lines layer loss ratio. Since the loss ratio is
\(\le 1\) there must be a positive layer margin density whenever
\(\beta_i(a)/\alpha_i(a) > 1\). But when \(\beta_i(a)/\alpha_i(a) < 1\)
it is possible the line has a negative margin density. How can that
occur and why does it make sense? To explore this we look at the shape
of \(\alpha\) and \(\beta\) in more detail.

It is important to remember why \cref{prop:pos-margin} does not apply:
it assumes unlimited cover, whereas here \(a<\infty\). With finite
capital there are potential transfers between lines caused by their
behavior in default that overwhelm the positive margin implied by the
proposition. Also note the proposition cannot be applied to
\(X\wedge a=\sum_i X_i(a)\) because the line payments are no longer
independent.

In general we can make two predictions about margins.

\textbf{Prediction 1}: Lines where \(\alpha_i(x)\) or \(\kappa_i(x)/x\)
increase with \(x\) will have always have a positive margin.

\textbf{Prediction 2}: A log-concave (thin tailed) line aggregated with
a non-log-concave (thick tailed) line can have a negative margin,
especially for lower asset layers.

Prediction 1 follows because the risk adjustment puts more weight on
\(X_i/X\) for larger \(X\) and so
\(\beta_i(x)/\alpha_i(x)> 1 > S(x) / g(S(x))\). Recall the risk
adjustment is comonotonic with total losses \(X\).

A thin tailed line aggregated with thick tailed lines will have
\(\alpha_i(x)\) decreasing with \(x\). Now the risk adjustment will
produce \(\beta_i(x)<\alpha_i(x)\) and it is possible that
\(\beta_i(x)/\alpha_i(x)<S(x)/g(S(x))\). In most cases, \(\alpha_i(x)\)
approaches \(\E[X_i]/x\) and \(\beta_i(x)/\alpha_i(x)\) increases with
\(x\), while the layer loss ratio decreases---and margin increases---and
the thin line will eventually get a positive margin. Whether or not the
thin line has a positive total margin \(\bar M_i(a)>0\) depends on the
particulars of the lines and the level of assets \(a\). A negative
margin is more likely for less well capitalized insurers, which makes
sense because default states are more material and they have a lower
overall dollar cost of capital. In the independent case, as
\(a\to\infty\) \cref{prop:pos-margin} eventually guarantees positive
margins for all lines.

These results are reasonable. Under limited liability, if assets and
liabilities are pooled then the thick tailed line benefits from pooling
with the thin one because pooling increases the assets available to pay
losses when needed. Equal priority transfers wealth from thin to thick
in states of the world where thick has a bad event, c.f., the example in
\cref{the-equal-priority-default-rule}. But because thick dominates the
total, the total losses are bad when thick is bad. The negative margin
compensates the thin-tailed line for transfers.

Another interesting situation occurs for asset levels within attritional
loss layers. Most realistic insured loss portfolios are quite skewed and
never experience very low loss ratios. For low loss layers, \(S(x)\) is
close to 1 and the layer at \(x\) is funded almost entirely by expected
losses; the margin and equity density components are nearly zero. Since
the sum of margin densities over component lines equals the total margin
density, when the total is zero it necessarily follows that either all
line margins are also zero or that some are positive and some are
negative. For the reasons noted above, thin tailed lines get the
negative margin as thick tailed lines compensate them for the improved
cover the thick tail lines obtain by pooling.

In conclusion, the natural margin by line reflects the relative
consumption of assets by layer, Mango
(\protect\hyperlink{ref-Mango2005a}{2005}). Low layers are less
ambiguous to the provider and have a lower margin relative to expected
loss. Higher layers are more ambiguous and have lower loss ratios. High
risk lines consume more higher layer assets and hence have a lower loss
ratio. For independent lines with no default the margin is always
positive. But there is a confounding effect when default is possible.
Because more volatile lines are more likely to cause default, there is a
wealth transfer to them. The natural premium allocation compensates low
risk policies for this transfer, which can result in negative margins in
some cases.

\hypertarget{equity-allocation-by-policy}{%
\section{Equity Allocation by
Policy}\label{equity-allocation-by-policy}}

Although \cref{eq:coc-by-line} determines margin by line, we cannot
compute return by line, or allocate frictional costs of capital, because
we still lack an equity allocation, a problem we now address.

\hypertarget{the-natural-allocation-of-equity}{%
\subsection{The Natural Allocation of
Equity}\label{the-natural-allocation-of-equity}}

\begin{definition}
The {\bf natural allocation of equity} to line $i$ is given by
\begin{equation}
Q_i(a) = \frac{\beta_i(a)g(S(a)) -  \alpha_i(x)S(a)}{g(S(a))- S(a)} \times (1-g(S(a))). \label{eq:main-alloc}
\end{equation}
\end{definition}

Why is this allocation natural? In total the layer return at \(a\) is \[
\iota(a) := \frac{M(a)}{Q(a)} = \frac{P(a) - S(a)}{1-P(a)} = \frac{g(S(a)) - S(a)}{1- g(S(a))}.
\] We claim that for a law invariant pricing measure the layer return
\emph{must be the same for all lines}. Law invariance implies the risk
measure is only concerned with the attachment probability of the layer
at \(a\), and not with the cause of loss within the layer. If return
\emph{within a layer} varied by line then the risk measure could not be
law invariant.

We can now compute capital by layer by line, by solving for the unknown
equity density \(Q_i(a)\) via \[
\iota(a) = \frac{M(a)}{Q(a)} = \frac{M_i(a)}{Q_i(a)}\implies Q_i(a) = \frac{M_i(a)}{\iota(a)}.
\] Substituting for layer return and line margin gives
\cref{eq:main-alloc}.

Since \(1-g(S(a))\) is the proportion of capital in the layer at \(a\),
\cref{eq:main-alloc} says the allocation to line \(i\) is given by the
nicely symmetric expression \begin{equation}\label{eq:q-formula}
\frac{\beta_i(a)g(S(a)) -  \alpha_i(x)S(a)}{g(S(a))- S(a)}.
\end{equation} To determine total capital by line we integrate the
equity density \[
\bar Q_i(a) := \int_0^a Q_i(x) dx.
\] And finally we can determine the average return to line \(i\) at
asset level \(a\) \begin{equation}\label{eq:avg-roe-by-line}
\bar\iota_i(a) = \frac{\bar M_i(a)}{\bar Q_i(a)}.
\end{equation} The average return will generally vary by line and by
asset level \(a\). Although the return within each layer is the same for
all lines, the margin, the proportion of capital, and the proportion
attributable to each line all vary by \(a\). Therefore average returns
will vary by line and \(a\). This is in stark contrast to the standard
industry approach, which uses the same return for each line and
implicitly all \(a\). How these quantities vary by line is complicated.
Academic approaches emphasized the possibility that returns vary by
line, but struggled with parameterization, Myers and Cohn
(\protect\hyperlink{ref-Myers1987}{1987}).

\Cref{eq:avg-roe-by-line} shows the average return by line is an
\(M_i\)-weighted harmonic mean of the layer returns given by the
distortion \(g\), viz \[
\frac{1}{\bar\iota_i(a)} = \int_0^a \frac{1}{\iota(x)}\frac{M_i(x)}{\bar M_i(a)}\,dx.
\] The harmonic mean solves the problem that the return for lower layers
of assets is potentially infinite (when \(g'(1)=0\)). The infinities do
not matter: at lower asset layers there is little or no equity and the
layer is fully funded by the loss component of premium. When so funded,
there is no margin and so the infinite return gets zero weight. In this
instance, the sense of the problem dictates that \(0\times\infty=0\):
with no initial capital there is no final capital regardless of the
return.

\hypertarget{intermediated-pricing}{%
\subsection{Intermediated Pricing}\label{intermediated-pricing}}

An equity allocation to policy is needed to price intermediated
insurance because of the frictional costs of holding capital in an
insurance company.

The price of investor-written insurance is \(\rho(X)\). A cat bond
transaction is an example of investor-written insurance. Following Myers
and Read Jr. (\protect\hyperlink{ref-Myers2001}{2001}) and Ibragimov,
Jaffee, and Walden (\protect\hyperlink{ref-Ibragimov2010}{2010}) we
model frictional costs as a tax on equity at rate \(\delta\). The
density and limited price of intermediated insurance becomes
\begin{gather}
P_i^I(a) = P_i(a) + \delta Q_i(a)  \\
\bar P_i^I(a) = \bar P_i(a) + \delta \bar Q_i(a). \\
\end{gather}

The relative size of \(M_i\) and \(\delta Q_i\) is a topic for future
research.

\hypertarget{examples}{%
\section{Examples}\label{examples}}

\hypertarget{example-1-a-simple-discrete-example}{%
\subsection{\texorpdfstring{Example 1: A Simple Discrete Example
\label{example-1}}{Example 1: A Simple Discrete Example }}\label{example-1-a-simple-discrete-example}}

Consider a two line example where \(X_1\) takes values 0, 9 and 10, and
\(X_2\) values 0, 1 and 90, the lines are independent and \(X=X_1+X_2\).
Suppose the outcome probabilities are \(1/2, 1/4\), and \(1/4\)
respectively for each outcome and consider a risk measure given by the
proportional hazard transform \(g(s)=\sqrt{s}\). There are nine possible
outcomes shown in \cref{tab:eg1}. The natural allocation appears to
depend on the ordering of the two outcomes where \(X=10\). If these two
rows are swapped the allocations are different, as shown in the last two
rows of the table.

\scriptsize

\begin{longtable}[]{@{}lrrrrrrrrrr@{}}
\caption{Nine possible outcomes showing ambiguous ordering for \(X=10\).
The natural allocation \(\E_{\mathsf{Q}}[X_i]\) appears to depend on the
ordering of outcomes 4 and 5. The next to last row shows
\(\E_{\mathsf{Q}}\) with these rows swapped.
\label{tab:eg1}}\tabularnewline
\toprule
Outcome & \(X_1\) & \(X_2\) & \(X\) & \(\mathsf P\) & \(S(x)\) &
\(g(S)\) & \(\mathsf{Q}\) & \(Z\) & \(\E[Z\mid X]\) &
\(\tilde{\mathsf{Q}}\)\tabularnewline
\midrule
\endfirsthead
\toprule
Outcome & \(X_1\) & \(X_2\) & \(X\) & \(\mathsf P\) & \(S(x)\) &
\(g(S)\) & \(\mathsf{Q}\) & \(Z\) & \(\E[Z\mid X]\) &
\(\tilde{\mathsf{Q}}\)\tabularnewline
\midrule
\endhead
1 & 0 & 0 & 0 & 4/16 & 12/16 & 0.8660254 & 0.1339746 & 0.5358984 &
0.5358984 & 0.1339746\tabularnewline
2 & 0 & 1 & 1 & 2/16 & 10/16 & 0.7905694 & 0.07545599 & 0.6036479 &
0.6036479 & 0.07545599\tabularnewline
3 & 9 & 0 & 9 & 2/16 & 8/16 & 0.7071068 & 0.08346263 & 0.6677011 &
0.6677011 & 0.08346263\tabularnewline
4 & 9 & 1 & 10 & 1/16 & 7/16 & 0.6614378 & 0.04566895 & 0.7307033 &
0.7898122 & 0.04936326\tabularnewline
5 & 10 & 0 & 10 & 2/16 & 5/16 & 0.5590170 & 0.1024208 & 0.8193667 &
0.7898122 & 0.09872652\tabularnewline
6 & 10 & 1 & 11 & 1/16 & 4/16 & 0.5 & 0.05901699 & 0.9442719 & 0.9442719
& 0.05901699\tabularnewline
7 & 0 & 90 & 90 & 2/16 & 2/16 & 0.3535534 & 0.1464466 & 1.171573 &
1.171573 & 0.1464466\tabularnewline
8 & 9 & 90 & 99 & 1/16 & 1/16 & 0.25 & 0.1035534 & 1.656854 & 1.656854 &
0.1035534\tabularnewline
9 & 10 & 90 & 100 & 1/16 & 0 & 0 & 0.25 & 4 & 4 & 0.25\tabularnewline
\(\E\) & 4.750000 & 22.75000 & 27.5 & & & & & & &\tabularnewline
\(\E_{\mathsf{Q}}\) & 6.208543 & 45.18014 & 51.38869 & & & & & &
&\tabularnewline
Swap 4, 5 & 6.200857 & 45.18783 & 51.38869 & & & & & & &\tabularnewline
Average & 6.2046998 & 45.183985 & 0 & & & & & & &\tabularnewline
\bottomrule
\end{longtable}

\normalsize

The last three columns of the table compute the measure
\(\tilde{\mathsf{Q}}\) corresponding to the unique \(\mathsf Q\mid X\).
The \(\tilde{\mathsf{Q}}\)-expected values of \(X_1\) and \(X_2\) are
6.2048488 and 45.183836, respectively. Note these values for the natural
allocation are different from the average of the the two orderings of
rows 4 and 5.

\Cref{tab:eg2} replaces \(X_i\) with \(\kappa_i(x)=\E[X_i\mid X=x]\),
resulting in one row per value of \(X,\) and uses \cref{main-theorem} to
compute expectations. The results are the same as using
\(\tilde{\mathsf{Q}}\).

\scriptsize

\begin{longtable}[]{@{}lrrrrrrr@{}}
\caption{Combining outcomes 4 and 5 and working with \(\E[X_i\mid X]\)
resolves the ambiguity and produces the natural allocation.
\label{tab:eg2}}\tabularnewline
\toprule
Outcome & \(E[X_1 | X]\) & \(E[X_2 | X]\) & \(X\) & \(\mathsf P\) &
\(S(x)\) & \(g(S)\) & \(\mathsf{Q}\)\tabularnewline
\midrule
\endfirsthead
\toprule
Outcome & \(E[X_1 | X]\) & \(E[X_2 | X]\) & \(X\) & \(\mathsf P\) &
\(S(x)\) & \(g(S)\) & \(\mathsf{Q}\)\tabularnewline
\midrule
\endhead
1 & 0 & 0 & 0 & 4/16 & 12/16 & 0.8660254 & 0.1339746\tabularnewline
2 & 0 & 1 & 1 & 2/16 & 10/16 & 0.7905694 & 0.07545599\tabularnewline
3 & 9 & 0 & 9 & 2/16 & 8/16 & 0.7071068 & 0.08346263\tabularnewline
4, 5 & 9 2/3 & 1 / 3 & 10 & 3/16 & 5/16 & 0.5590170 &
0.1480898\tabularnewline
6 & 10 & 1 & 11 & 1/16 & 4/16 & 0.5 & 0.05901699\tabularnewline
7 & 0 & 90 & 90 & 2/16 & 2/16 & 0.3535534 & 0.1464466\tabularnewline
8 & 9 & 90 & 99 & 1/16 & 1/16 & 0.25 & 0.1035534\tabularnewline
9 & 10 & 90 & 100 & 1/16 & 0 & 0 & 0.25\tabularnewline
& & & & & & &\tabularnewline
\(\E\) & 4.75 & 22.75 & 27.5 & & & &\tabularnewline
\(\E_{\mathsf{Q}}\) & 6.2048488 & 45.183836 & 51.38869 & & &
&\tabularnewline
\bottomrule
\end{longtable}

\normalsize

Actuaries commonly perform this type of calculation, often with
catastrophe model output. They make the simplifying assumption that
\(X_i=\E[X_i\mid X]\) when all rows are distinct. However, the ordering
problem illustrated does occur in real data, especially when limits and
retentions are involved. \Cref{main-theorem} shows how to rigorously
resolve the ordering problem to compute the unique natural allocation.

\hypertarget{example-2-thick-tailed-and-thin-tailed-lines}{%
\subsection{\texorpdfstring{Example 2: Thick-tailed and Thin-tailed
Lines
\label{example-2}}{Example 2: Thick-tailed and Thin-tailed Lines }}\label{example-2-thick-tailed-and-thin-tailed-lines}}

Example 2 is based on two distributions with mean 1. \(X_1\) is thin
tailed with a gamma distribution with coefficient of variation 0.25.
\(X_2\) is a translated lognormal \(X_2=0.3 + 0.7X'_2\), where \(X'_2\)
has a coefficient of variation \(1.25/0.7\), resulting in a coefficient
of variation of \(1.25\) for \(X_2\). Total assets are 12.5,
corresponding to capital at a 563 year return period. The aggregate
coefficient of variation is 0.637 in total. \(X_1\) approximates a
moderate limit book of commercial auto and \(X_2\) a catastrophe exposed
property book with a stable attritional loss component. In aggregate the
portfolio would be considered as volatile. The distortion \(g\) uses a
Wang transform with \(\lambda=0.755\), producing a 10 percent return on
assets. The natural allocation premium is 1.057 for \(X_1\) and 1.889
for line (94.6 percent and 52.4 percent loss ratios), producing an
overall 67.6 percent loss ratio, all without expenses. The profit is
realistic for a gross portfolio with these characteristics.

\Cref{fig:eg-2} illustrates the theory we have developed. We refer to
the charts as \((r,c)\) for row \(r=1,2,3,4\) and column \(c=1,2,3\),
starting at the top left. The horizontal axis shows the asset level in
all charts except \((3,3)\) and \((4,3)\), where it shows probability,
and \((1,3)\) where it shows loss. Blue represents the thin tailed line,
orange thick tailed and green total. When both dashed and solid lines
appear on the same plot, the solid represent risk-adjusted and dashed
represent non-risk-adjusted functions. Here is the key.

\begin{itemize}
\tightlist
\item
  \((1,1)\) shows density for \(X_1, X_2\) and \(X=X_1+X_2\); the two
  lines are independent. Both lines have mean 1.
\item
  \((1,2)\): log density; comparing tail thickness.
\item
  \((1,3)\): the bivariate log-density. This plot illustrates where
  \((X_1, X_2)\) \emph{lives}. The diagonal lines show \(X=k\) for
  different \(k\). These show that large values of \(X\) correspond to
  large values of \(X_2\), with \(X_1\) about average.
\item
  \((2,1)\): the form of \(\kappa_i\) is clear from looking at
  \((1,3)\). \(\kappa_1\) peaks at \(x=2.15\) with maximum value 1.14.
  Thereafter it declines to 1.0. \(\kappa_2\) is monotonically
  increasing.
\item
  \((2,2)\): The \(\alpha_i\) functions. For small \(x\) the expected
  proportion of losses is approximately 50/50, since the means are
  equal. As \(x\) increases \(X_2\) dominates. The two functions sum to
  1.
\item
  \((2,3)\): The solid lines are \(\beta_i\) and the dashed lines
  \(\alpha_i\) from \((2,2)\). Since \(\alpha_1\) decreases
  \(\beta_1(x)\le \alpha_1(x)\). This can lead to \(X_1\) having a
  negative margin in low asset layers. \(X_2\) is the opposite.
\item
  \((3,1)\): illustrates premium and margin determination by asset layer
  for \(X_1\) using \cref{eq:alpha-S} and \cref{eq:beta-gS}. For low
  asset layers \(\alpha_1(x) S(x)>\beta_1(x) g(S(x))\) (dashed above
  solid) corresponding to a negative margin. Beyond about \(x=1.38\) the
  lines cross and the margin is positive.
\item
  \((4,1)\): shows the same thing for \(X_2\). Since \(\alpha_2\) is
  increasing, \(\beta_2(x)>\alpha_2(x)\) for all \(x\) and so all layers
  get a positive margin. The solid line \(\beta_2 gS\) is above the
  dashed \(\alpha_2 S\) line.
\item
  \((3,2)\): the layer margin densities. For low asset layers premium is
  fully funded by loss with zero overall margin. \(X_2\) requires a
  positive margin and \(X_1\) a negative one, reflecting the benefit the
  thick line receives from pooling in low layers. The overall margin is
  always non-negative. Beyond \(x=1.38\), \(X_1\)'s margin is also
  positive.
\item
  \((4,2)\): the cumulative margin in premium by asset level. Total
  margin is zero in low \emph{dollar-swapping} layers and then
  increases. It is always non-negative. The curves in this plot are the
  integrals of those in \((3,2)\) from 0 to \(x\).
\item
  \((3,3)\): shows stand-alone loss \((1-S(x),x)=(p,q(p))\) (dashed) and
  premium \((1-g(S(x)),x)=(p,q(1-g^{-1}(1-p))\) (solid, shifted left)
  for each line and total. The margin is the shaded area between the
  two. Each set of three lines (solid or dashed) does not add up
  vertically because of diversification. The same distortion \(g\) is
  applied to each line to the stand-alone \(S_{X_i}\). It is calibrated
  to produce a 10 percent return overall. On a stand-alone basis,
  calculating capital by line to the same return period as total,
  \(X_1\) is priced to a 83.5 percent loss ratio and \(X_2\) a 51.8
  percent, producing an average 64.0 percent, vs.~67.6 percent on a
  combined basis. Returns are 28.7 percent and 9.6 percent respectively,
  averaging 10.9 percent, vs 10 percent on a combined basis.
\item
  \((4,3)\): shows the natural allocation of loss and premium to each
  line. The total (green) is the same as \((3,3)\). For each \(i\),
  dashed shows \((p, \E[X_i\mid X=q(p)])\), i.e.~the expected loss
  recovery conditioned on total losses \(X=q(p)\), and solid shows
  \((p, \E[X_i\mid X=q(1-g^{-1}(1-p))])\), i.e.~the natural premium
  allocation (see the bottom row of
  \cref{tab:natural-allocation-summary}). Here the solid and dashed
  lines \emph{add up} vertically: they are allocations of the total.
  Looking vertically above \(p\), the shaded areas show how the total
  margin at that loss level is allocated between lines. \(X_1\) mostly
  consumes assets at low layers, and the blue area is thicker for small
  \(p\), corresponding to smaller total losses. For \(p\) close to 1,
  large total losses, margin is dominated by \(X_2\) and in fact \(X_1\)
  gets a slight credit (dashed above solid). The change in shape of the
  shaded margin area for \(X_1\) is particularly evident: it shows
  \(X_1\) benefits from pooling and requires a lower overall margin. The
  natural allocation returns are 5.3, 10.6 and 10.0 percent. The overall
  premium to surplus leverage is 0.308 to 1; on an allocated basis it is
  0.986 and 0.223 to 1 for each line.
\end{itemize}

There may appear to be a contradiction between figures \((3,2)\) and
\((4,3)\) but it should be noted that a particular \(p\) value in
\((4,3)\) refers to different events on the dotted and solid lines.

Plots \((3,3)\) and \((4,3)\) explain why the thick line requires
relatively more margin (0.698 out of a total 0.728): its shape does not
change when it is pooled with \(X_1\). In \((3,3)\) the green shaded
area is essentially an upwards shift of the orange, and the orange areas
in \((3,3)\) and \((3,4)\) are essentially the same. This means that
adding \(X_1\) has virtually no impact on the shape of \(X_2\); it is
like adding a constant, as discussed after the proof of
\cref{prop:pos-margin}. This can also be seen in \((4,3)\) where the
blue region is almost a straight line.

These shifts are illustrated in \cref{fig:eg-2-detail}. The left hand
plot shows that the stand-alone margin area for \(X_2\) shifted up by 1,
the mean of \(X_1\), lies almost exactly over the total margin area
(orange over green). The right hand plot compares the stand-alone margin
areas for each line with the natural margin allocation. \(X_2\) is
shifted up by 1 for clarity. Again, there is essentially no difference
for \(X_2\), especially in the expensive, large loss states where \(p\)
is close to 1. \(X_1\) is completely transformed: its margin is much
lower (smaller area) and it is concentrated in low \(p\), small total
loss events. \(X_1\) actually gets a credit at large total losses
because its losses will be close to the mean, and hence low ambiguity,
whereas \(X_2\)'s loss will be large and more ambiguous.

\begin{figure}
\centering
\includegraphics[width=1\textwidth,height=\textheight]{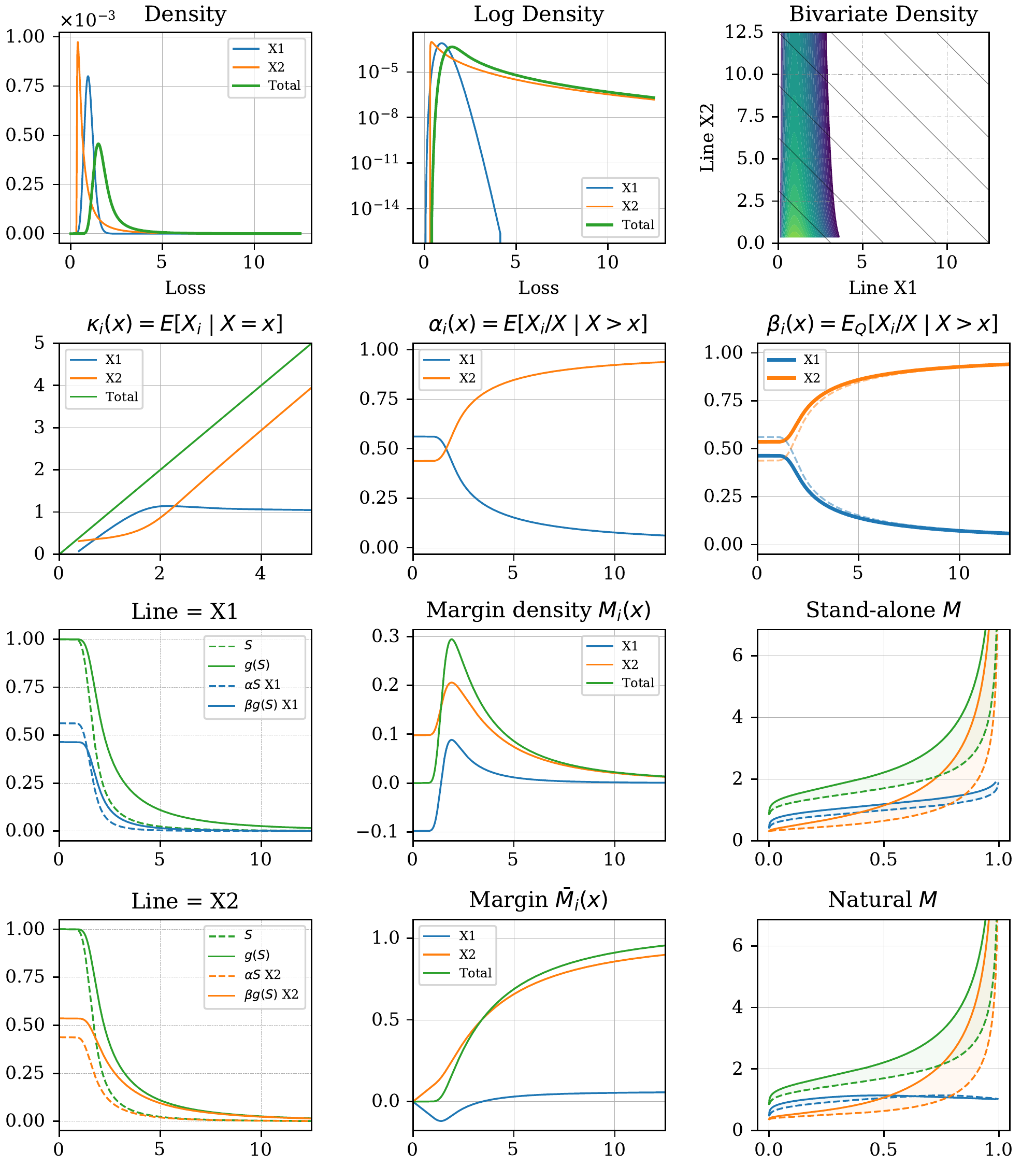}
\caption{A thin tailed line combined with a thick tailed line. See text
for a key to the graphs. \label{fig:eg-2}}
\end{figure}

\begin{figure}
\centering
\includegraphics{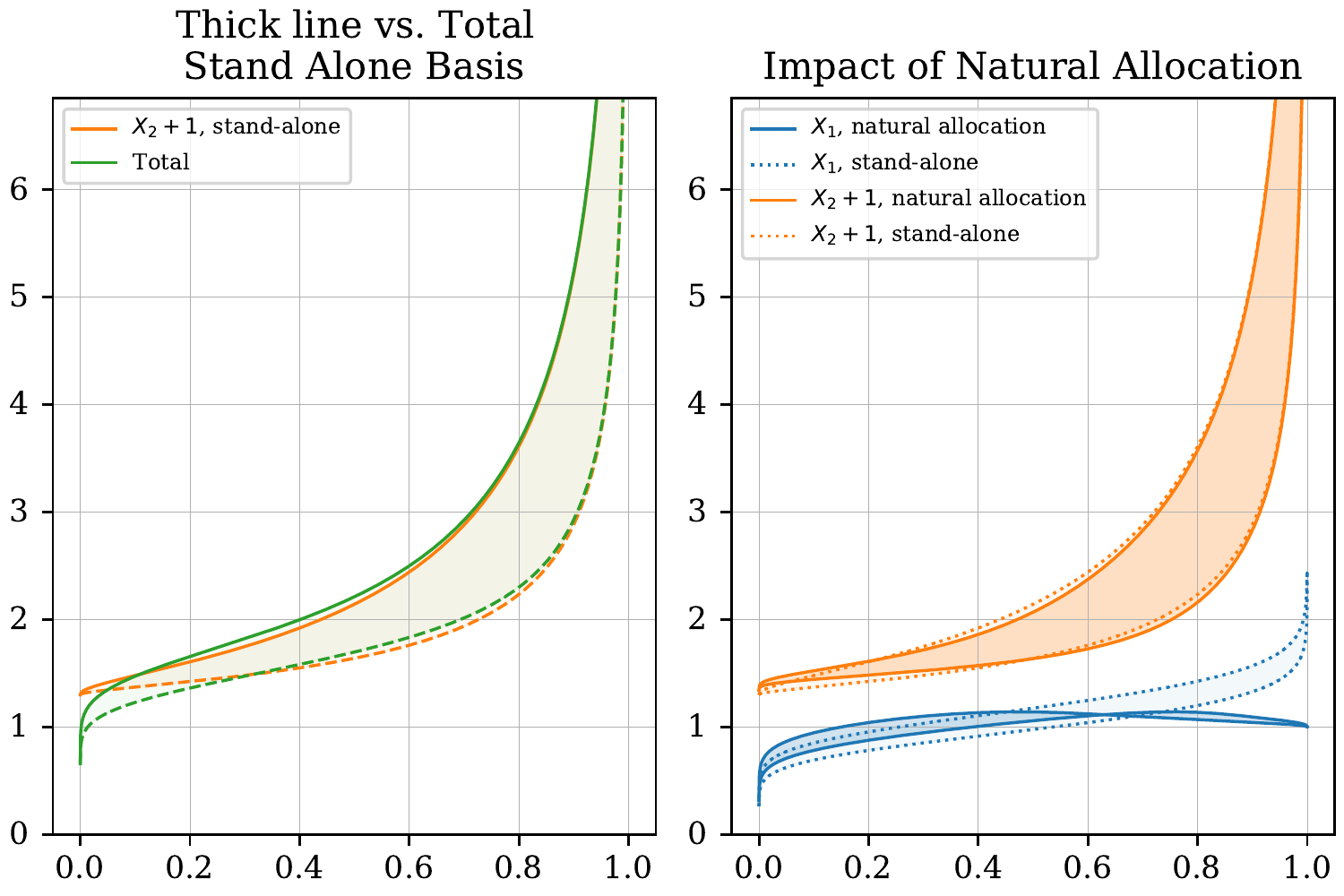}
\caption{Impact of the natural allocation by line. Left: stand-alone
margin area for thick line shifted up by mean of thin line lies almost
exactly over the margin area for the total. Right: stand-alone
vs.~natural allocation margin areas, showing minimal impact for thick
line but dramatic impact on thin. \label{fig:eg-2-detail}}
\end{figure}

\hypertarget{conclusions}{%
\section{Conclusions}\label{conclusions}}

We have explored how the shape of risk impacts the price of risk
transfer in an imperfect, incomplete market---a holy grail for
practicing actuaries. We provide a natural, assumption-free allocation
of aggregate premium to policy, incorporating an allocation of capital
consistent with law invariance in order to price individual policies in
the presence of frictional costs. Premium by policy is determined by the
relative consumption of low and high ambiguity assets in a complex, but
intuitively reasonable manner. The margin by line is driven more by
behavior in solvent states than in default states (default states are
often the major focus of allocation methods). Premium is interpreted as
the value of insurance cash flows under a risk-dependent state price
density. This is in contrast to most other approaches that adopt a cost
allocation perspective. The natural allocation is insured-centric,
rather than insurer-centric.

The natural allocation depends on the fact that DRMs are law invariant
and comonotonic additive. It does not apply to more general convex risk
measures. Notwithstanding this limitation, it provides a useful and
practical method that can be applied by an insurance company to
understand how to share its diversification benefit between policies.

Further research is needed to determine how the shapes of the \(X_i\)
interact to determine the natural allocation, as well as the impact of
different distortions \(g\) and capitalization standards. The underlying
distortion can be calibrated to market prices. Market calibration to cat
bond data and standard intermediated insurance data would reveal how
much of the cost of capital arises from frictional costs and how much
from shape of risk. The work can also be extended to dynamic portfolios
and then applied to questions of optimal risk pooling under costly
capital.

\hypertarget{references}{%
\section*{References}\label{references}}
\addcontentsline{toc}{section}{References}

\hypertarget{refs}{}
\begin{cslreferences}
\leavevmode\hypertarget{ref-Acerbi2002b}{}%
Acerbi, Carlo, 2002, Spectral measures of risk: A coherent
representation of subjective risk aversion, \emph{Journal of Banking \&
Finance} 26, 1505--1518.

\leavevmode\hypertarget{ref-Anscombe1963}{}%
Anscombe, F. J., and R. J. Aumann, 1963, A Definition of Subjective
Probability, \emph{The Annals of Mathematical Statistics} 34, 199--205.

\leavevmode\hypertarget{ref-Bauer2013}{}%
Bauer, Daniel, and George H Zanjani, 2013, Capital Allocation and its
Discontents, \emph{Handbook of insurance} (Springer, New York, NY).

\leavevmode\hypertarget{ref-Bauerle2006}{}%
Bäuerle, Nicole, and Alfred Müller, 2006, Stochastic orders and risk
measures: Consistency and bounds, \emph{Insurance: Mathematics and
Economics} 38, 132--148.

\leavevmode\hypertarget{ref-Bodoff2007}{}%
Bodoff, Neil M., 2007, Capital Allocation by Percentile Layer,
\emph{Variance} 3, 13--30.

\leavevmode\hypertarget{ref-Tsanakas2016a}{}%
Boonen, Tim J., Andreas Tsanakas, and Mario V. Wüthrich, 2017, Capital
allocation for portfolios with non-linear risk aggregation,
\emph{Insurance: Mathematics and Economics} 72, 95--106.

\leavevmode\hypertarget{ref-Borch1982}{}%
Borch, Karl, 1982, Additive Insurance Premiums: A Note, \emph{The
Journal of Finance} 37, 1295--1298.

\leavevmode\hypertarget{ref-Borwein2010}{}%
Borwein, Jonathan M, and Jon D Vanderwerff, 2010, \emph{Convex Functions
- Construction, Characterizations and Counterexamples} (Cambridge
University Press).

\leavevmode\hypertarget{ref-Campi2013}{}%
Campi, Luciano, Elyès Jouini, and Vincent Porte, 2013, Efficient
portfolios in financial markets with proportional transaction costs,
\emph{Mathematics and Financial Economics} 7, 281--304.

\leavevmode\hypertarget{ref-Carlier2003}{}%
Carlier, G., and R. A. Dana, 2003, Core of convex distortions of a
probability, \emph{Journal of Economic Theory} 113, 199--222.

\leavevmode\hypertarget{ref-Castagnoli2002}{}%
Castagnoli, Erio, Fabio Maccheroni, and Massimo Marinacci, 2002,
Insurance premia consistent with the market, \emph{Insurance:
Mathematics and Economics} 31, 267--284.

\leavevmode\hypertarget{ref-Castagnoli2004}{}%
Castagnoli, Erio, Fabio Maccheroni, and Massimo Marinacci, 2004, Choquet
insurance pricing: A caveat, \emph{Mathematical Finance} 14, 481--485.

\leavevmode\hypertarget{ref-Chateauneuf1996a}{}%
Chateauneuf, A., R. Kast, and A. Lapied, 1996, Choquet pricing for
financial markets with frictions, \emph{Mathematical Finance} 6,
323--330.

\leavevmode\hypertarget{ref-Culp2009}{}%
Culp, Christopher L., and Kevin J. O'Donnell, 2009, Catastrophe
reinsurance and risk capital in the wake of the credit crisis, \emph{The
Journal of Risk Finance} 10, 430--459.

\leavevmode\hypertarget{ref-Cummins1988}{}%
Cummins, J. David, 1988, Risk-Based Premiums for Insurance Guarantee
Funds, \emph{Journal of Finance} 43, 823--839.

\leavevmode\hypertarget{ref-Cummins1990a}{}%
Cummins, J. David, 1990, Multi-Period Discounted Cash Flow Rate-making
Models in Property-Liability Insurance, \emph{Journal of Risk and
Insurance} 57, 79--109.

\leavevmode\hypertarget{ref-Cummins2000}{}%
Cummins, J. David, 2000, Allocation of capital in the insurance
industry, \emph{Risk Management and Insurance Review}, 7--28.

\leavevmode\hypertarget{ref-Cummins2005}{}%
Cummins, J. David, and Richard D. Phillips, 2005, Estimating the Cost of
Equity Capital for Property-Liability Insurers, \emph{Journal of Risk
and Insurance} 72, 441--478.

\leavevmode\hypertarget{ref-Doherty1990a}{}%
D'Arcy, Stephen P., and Neil A Doherty, 1990, Adverse Selection, Private
Information, and Lowballing in Insurance Markets, \emph{The Journal of
Business} 63, 145.

\leavevmode\hypertarget{ref-Delbaen2000}{}%
Delbaen, Freddy, 2000, Coherent risk measures (Pisa Notes),
\emph{Blätter der DGVFM} 24, 733--739.

\leavevmode\hypertarget{ref-Delbaen1994}{}%
Delbaen, Freddy, and Walter Schachermayer, 1994, A general version of
the fundamental theorem of asset pricing, \emph{Mathematische Annalen}
300, 463--520.

\leavevmode\hypertarget{ref-Denault2001}{}%
Denault, Michel, 2001, Coherent allocation of risk capital, \emph{The
Journal of Risk} 4, 1--34.

\leavevmode\hypertarget{ref-Denneberg1994}{}%
Denneberg, Dieter, 1994, \emph{Non-additive measure and integral}
(Kluwer Academic, Dordrecht).

\leavevmode\hypertarget{ref-Denuit2012}{}%
Denuit, Michel, and Jan Dhaene, 2012, Convex order and comonotonic
conditional mean risk sharing, \emph{Insurance: Mathematics and
Economics} 51, 265--270.

\leavevmode\hypertarget{ref-Denuit2020e}{}%
Denuit, Michel M., and C Y Robert, 2020, Risk Reduction by Conditional
Mean Risk Sharing With Application to Collaborative Insurance,.
Discussion paper (UC Louvain).

\leavevmode\hypertarget{ref-DeWaegenaere2000}{}%
De Waegenaere, Anja, 2000, Arbitrage and Viability in Insurance Markets,
\emph{GENEVA Papers on Risk and Insurance Theory} 25, 81--99.

\leavevmode\hypertarget{ref-DeWaegenaere2003}{}%
De Waegenaere, Anja, Robert Kast, and Andre Lapied, 2003, Choquet
pricing and equilibrium, \emph{Insurance: Mathematics and Economics} 32,
359--370.

\leavevmode\hypertarget{ref-Dhaene2012b}{}%
Dhaene, Jan, Alexander Kukush, Daniel Linders, and Qihe Tang, 2012,
Remarks on quantiles and distortion risk measures, \emph{European
Actuarial Journal} 2, 319--328.

\leavevmode\hypertarget{ref-Dietz2017}{}%
Dietz, Simon, and Oliver Walker, 2017, Ambiguity and Insurance: Capital
Requirements and Premiums, \emph{Journal of Risk and Insurance} 86,
213--235.

\leavevmode\hypertarget{ref-Doherty1986}{}%
Doherty, Neil A., and James R. Garven, 1986, Price Regulation in
Property-Liability Insurance: A Contingent-Claims Approach,
\emph{Journal of Finance} 41, 1031--1050.

\leavevmode\hypertarget{ref-Dybvig1989}{}%
Dybvig, Philip H, and Stephen A Ross, 1989, Arbitrage, \emph{Finance}
(Palgrave Macmillan, London).

\leavevmode\hypertarget{ref-Epstein2008}{}%
Epstein, Larry G, and Martin Schneider, 2008, Ambiguity, Information
Quality, and Asset Pricing, \emph{Journal of Finance} LXIII, 197--228.

\leavevmode\hypertarget{ref-Follmer2011}{}%
Föllmer, Hans, and Alexander Schied, 2011, \emph{Stochastic finance: an
introduction in discrete time}. Third Edit. (Walter de Gruyter).

\leavevmode\hypertarget{ref-Froot2008}{}%
Froot, Kenneth A., and Paul G J O'Connell, 2008, On the pricing of
intermediated risks: Theory and application to catastrophe reinsurance,
\emph{Journal of Banking and Finance} 32, 69--85.

\leavevmode\hypertarget{ref-Grundl2007}{}%
Grundl, Helmut, and Hato Schmeiser, 2007, Capital allocation for
insurance companies---What Good Is It?, \emph{Journal of Risk and
Insurance} 74.

\leavevmode\hypertarget{ref-Huber1981}{}%
Huber, Peter J., 1981, \emph{Robust statistics} (John Wiley \& Sons,
Inc.).

\leavevmode\hypertarget{ref-Ibragimov2010}{}%
Ibragimov, Rustam, Dwight Jaffee, and Johan Walden, 2010, Pricing and
Capital Allocation for Multiline Insurance Firms, \emph{Journal of Risk
and Insurance} 77, 551--578.

\leavevmode\hypertarget{ref-Jiang2020}{}%
Jiang, Wenjun, Marcos Escobar-Anel, and Jiandong Ren, 2020, Optimal
insurance contracts under distortion risk measures with ambiguity
aversion, \emph{ASTIN Bulletin} 50, 1--28.

\leavevmode\hypertarget{ref-Jouini2001}{}%
Jouini, Elyès, and Hédi Kallal, 2001, Efficient trading strategies in
the presence of market frictions, \emph{Review of Financial Studies} 14,
343--369.

\leavevmode\hypertarget{ref-Klibanoff2005}{}%
Klibanoff, Peter, Massimo Marinacci, and Sujoy Mukerji, 2005, A smooth
model of decision making under ambiguity, \emph{Econometrica} 73,
1849--1892.

\leavevmode\hypertarget{ref-Kusuoka2001}{}%
Kusuoka, Shigeo, 2001, On law invariant coherent risk measures,
\emph{Advances in Mathematical Economics} 3, 83--95.

\leavevmode\hypertarget{ref-Machina1997}{}%
Machina, Mark J., and John W. Pratt, 1997, Increasing Risk: Some Direct
Constructions, \emph{Journal of Risk and Uncertainty} 14, 103--127.

\leavevmode\hypertarget{ref-Major2018}{}%
Major, John A., 2018, Distortion Measures on Homogeneous Financial
Derivatives, \emph{Insurance: Mathematics and Economics} 79, 82--91.

\leavevmode\hypertarget{ref-Mango2005a}{}%
Mango, Donald, 2005, Insurance Capital as a Shared Asset, \emph{Astin
Bulletin} 35, 471--486.

\leavevmode\hypertarget{ref-Mango2013}{}%
Mango, Donald, John Major, Avraham Adler, and Claude Bunick, 2013,
Capital Tranching: A RAROC Approach to Assessing Reinsurance Cost
Effectiveness, \emph{Variance} 7, 82--91.

\leavevmode\hypertarget{ref-Meyers1996}{}%
Meyers, Glenn G, 1996, The competitive market equilibrium risk load
formula for catastrophe ratemaking, \emph{PCAS}, 563--600.

\leavevmode\hypertarget{ref-Mildenhall2017b}{}%
Mildenhall, Stephen J, 2017, Actuarial Geometry, \emph{Risks} 5.

\leavevmode\hypertarget{ref-Myers1987}{}%
Myers, Stewart C, and Richard A Cohn, 1987, A discounted cash flow
approach to property-liability insurance rate regulation, \emph{Fair
rate of return in property-liability insurance} (Springer).

\leavevmode\hypertarget{ref-Myers2001}{}%
Myers, Stewart C, and James A Read Jr., 2001, Capital allocation for
insurance companies, \emph{Journal of Risk and Insurance} 68, 545--580.

\leavevmode\hypertarget{ref-Phillips1998}{}%
Phillips, Richard D., J. David Cummins, and Franklin Allen, 1998,
Financial Pricing of Insurance in the Multiple-Line Insurance Company,
\emph{Journal of Risk and Insurance} 65, 597--636.

\leavevmode\hypertarget{ref-Robert2014}{}%
Robert, Christian Y., and Pierre-E. Therond, 2014, Distortion Risk
Measures, Ambiguity Aversion and Optimal Effort, \emph{ASTIN Bulletin}
44, 277--302.

\leavevmode\hypertarget{ref-Ross1978}{}%
Ross, Stephen A, 1978, A Simple Approach to the Valuation of Risky
Streams, \emph{The Journal of Business} 51, 453.

\leavevmode\hypertarget{ref-Rothschild1970}{}%
Rothschild, Michael, and Joseph E. Stiglitz, 1970, Increasing risk: I. A
definition, \emph{Journal of Economic Theory} 2, 225--243.

\leavevmode\hypertarget{ref-Saumard2014}{}%
Saumard, Adrien, and Jon A. Wellner, 2014, Log-concavity and strong
log-concavity: a review, \emph{Statistics Surveys} 8, 45--114.

\leavevmode\hypertarget{ref-Schmeidler1986}{}%
Schmeidler, David, 1986, Integral representation without additivity,
\emph{Proceedings of the American Mathematical Society} 97, 255--255.

\leavevmode\hypertarget{ref-Schmeidler1989}{}%
Schmeidler, David, 1989, Subjective Probability and Expected Utility
without Additivity, \emph{Econometrica} 57, 571--587.

\leavevmode\hypertarget{ref-Shapiro2009}{}%
Shapiro, Alexander, Darinka Dentcheva, and Andrzej Ruszczyński, 2009,
\emph{Lectures on Stochastic Programming}. May.

\leavevmode\hypertarget{ref-Sherris2006a}{}%
Sherris, Michael, 2006, Solvency, capital allocation, and fair rate of
return in insurance, \emph{Journal of Risk and Insurance} 73, 71--96.

\leavevmode\hypertarget{ref-Svindland2010}{}%
Svindland, Gregor, 2009, Continuity properties of law-invariant
(quasi-)convex risk functions on \(L^\infty\), \emph{Mathematics and
Financial Economics} 3, 39--43.

\leavevmode\hypertarget{ref-Tasche1999}{}%
Tasche, Dirk, 1999, Risk contributions and performance measurement,
\emph{Report of the Lehrstuhl fur mathematische Statistik, TU Munchen},
1--26.

\leavevmode\hypertarget{ref-Thaler1998}{}%
Thaler, Richard H, 1988, Anomalies: The Winner's Curse, \emph{Journal of
Economic Perspectives} 2, 191--202.

\leavevmode\hypertarget{ref-Tsanakas2003a}{}%
Tsanakas, Andreas, and Christopher Barnett, 2003, Risk capital
allocation and cooperative pricing of insurance liabilities,
\emph{Insurance: Mathematics and Economics} 33, 239--254.

\leavevmode\hypertarget{ref-Venter1991}{}%
Venter, Gary G., 1991, Premium Calculation Implications of Reinsurance
Without Arbitrage, \emph{ASTIN Bulletin} 21, 223--230.

\leavevmode\hypertarget{ref-Venter2006}{}%
Venter, Gary G., John A. Major, and Rodney E. Kreps, 2006, Marginal
Decomposition of Risk Measures, \emph{ASTIN Bulletin} 36, 375--413.

\leavevmode\hypertarget{ref-Wang1995}{}%
Wang, Shaun, 1995, Insurance pricing and increased limits ratemaking by
proportional hazards transforms, \emph{Insurance: Mathematics and
Economics} 17, 43--54.

\leavevmode\hypertarget{ref-Wang1996}{}%
Wang, Shaun, 1996, Premium Calculation by Transforming the Layer Premium
Density, \emph{ASTIN Bulletin} 26, 71--92.

\leavevmode\hypertarget{ref-Wang1997}{}%
Wang, Shaun, Virginia Young, and Harry Panjer, 1997, Axiomatic
characterization of insurance prices, \emph{Insurance: Mathematics and
Economics} 21, 173--183.

\leavevmode\hypertarget{ref-Yaari1987}{}%
Yaari, Menahem E., 1987, The Dual Theory of Choice under Risk, \emph{The
Econometric Society} 55, 95--115.

\leavevmode\hypertarget{ref-Zhang2002}{}%
Zhang, Jiankang, 2002, Subjective Ambiguity, Expected Utility and
Choquet Expected Utility, \emph{Economic Theory} 20, 159--181.
\end{cslreferences}

\end{document}